\documentclass [10pt]{amsart}
\usepackage{amsmath}
\usepackage{amssymb}
\usepackage{amscd}
\usepackage{epsfig}
\usepackage{xcolor}
\usepackage{amsbsy}
\usepackage{scalerel,stackengine}

\usepackage[
colorlinks,
citecolor=red,
urlcolor=blue,
linkcolor=red,
allcolors=black,
backref=page
]{hyperref}

\stackMath
\newcommand\reallywidehat[1]{%
\savestack{\tmpbox}{\stretchto{%
  \scaleto{%
    \scalerel*[\widthof{\ensuremath{#1}}]{\kern-.6pt\bigwedge\kern-.6pt}%
    {\rule[-\textheight/2]{1ex}{\textheight}}%WIDTH-LIMITED BIG WEDGE
  }{\textheight}%
}{0.5ex}}%
\stackon[1pt]{#1}{\tmpbox}%
}

\newcommand\reallywidecheck[1]{%
\savestack{\tmpbox}{\stretchto{%
  \scaleto{%
    \scalerel*[\widthof{\ensuremath{#1}}]{\kern-.6pt\bigwedge\kern-.6pt}%
    {\rule[-\textheight/2]{1ex}{\textheight}}%WIDTH-LIMITED BIG WEDGE
  }{\textheight}%
}{0.5ex}}%
\stackon[1pt]{#1}{\scalebox{-1}{\tmpbox}}%
}

\numberwithin{equation}{section}

\allowdisplaybreaks

\newcommand{\card}{\mbox{\rm card}}

\newcommand{\supp}{\mbox{\rm supp}}

\newcommand{\RR}{{\mathbb R}}
\newcommand{\QQ}{{\mathbb Q}}
\newcommand{\ZZ}{{\mathbb Z}}
\newcommand{\CC}{{\mathbb C}}
\newcommand{\AAA}{{\mathbb A}}
\newcommand{\TT}{\mathbb T}
\newcommand{\NN}{\mathbb N}

\newcommand{\XX}{\mathbb X}

\newcommand{\cL}{{\mathcal L}}

\newcommand{\mm}{\mathsf{m}}

\newcommand{\ee}{\mathsf{e}}

\newcommand{\dd}{\mbox{\rm d}}

 \newtheorem{theorem}{Theorem}[section]
 \newtheorem{lemma}[theorem]{Lemma}
 \newtheorem{proposition}[theorem]{Proposition}
 \newtheorem{corollary}[theorem]{Corollary}

\theoremstyle{definition}
 
 \newtheorem{conj}[theorem]{Conjecture}
 \newtheorem{quest}[theorem]{Question}
 \newtheorem{definition}[theorem]{Definition}
 
 \newtheorem{remark}[theorem]{Remark}

\begin{document}

\title{Pseudorandomness and Diffraction}

\author[D.\ Damanik]{David Damanik}

\address{Department of Mathematics, Rice University, Houston, TX~77005, USA}

\email{damanik@rice.edu}

\thanks{D.\ D.\ was supported in part by NSF grant DMS--2349919. N.\ S.\ was supported by NSERC grant 2024-04853. }

\author[N.\ Strungaru]{Nicolae Strungaru}

\address{Department of Mathematics and Statistics, MacEwan University, 10700 -- 104 Avenue, Edmonton, AB, T5J 4S2, Canada}

\email{strungarun@macewan.ca}


\begin{abstract}
Pseudorandom structures are structures that behave like random ones, without necessarily being random themselves. It is a prominent, and evidently very difficult, conjecture that
$$
V(n) = \lambda \cos \Big( 2 \pi \Big( x_1 + n x_2 + \frac{n(n-1)}{2} \alpha \Big) \Big)
$$
is pseudorandom from a Schr\"odinger operator perspective in that the Schr\"odinger operator in $\ell^2(\ZZ)$ with potential $V$ displays Anderson localization, that is, pure point spectrum with exponentially decaying eigenfunctions for almost all parameter values --- the same spectral features as those produced by random potentials. We show that $V$ is pseudorandom in terms of its diffraction properties, that is, the associated diffraction measure is purely absolutely continuous for all $\lambda \not= 0$, all irrational $\alpha$, and all $x_1,x_2$ --- which is the case as well for the random case. Our result gives further evidence for the conjecture in the Schr\"odinger case and it elucidates the apparent dual behavior of Schr\"odinger spectral measures and diffraction measures.
\end{abstract}

\maketitle

%\tableofcontents

\section{Introduction and Motivation}

The study of order/disorder properties of structures goes back a long time. Most relevant to the present discussion are the Nobel Prize winning works on the phenomenon of Anderson localization in 1958 by P.~Anderson \cite{A58} (1977 Nobel Prize in Physics) and the discovery of quasicrystals in 1984 by D.~Shechtman \cite{SBGC84} (2011 Nobel Prize in Chemistry).

Accordingly, we will consider two ways of studying signatures of order/disorder properties.

Anderson studied the localization of quantum states by random environments, and hence the mathematical objects to consider are Schr\"odinger operators with random potentials in suitable $\ell^2$ spaces. Localization features are reflected on a spectral level by such operators having pure point spectral measures with exponentially decaying eigenfunctions. For textbook treatments of Anderson localization, see for example \cite{AW15, CS14, DF24, DKKKR08, S01}.

Shechtman discovered structures whose diffraction is characterized by having bright spots (so-called \emph{Bragg peaks}) along with rotational symmetries that are impossible for lattices in two or three dimensions. Thus, these structures share crucial features with crystals (which display Bragg peaks in their diffraction as well), while lacking spatial periodicity (as classically required of crystals). They are therefore typically referred to as \emph{quasicrystals}. Mathematically, diffraction is modeled via the so-called \emph{diffraction measure}, and its pure point nature is reflective of the order properties of the underlying structure. For an introduction to diffraction, and more generally aperiodic order, see for example \cite{TAO}.

Thus, in these two settings, measures are associated with structures, the spectral measures in the Schr\"odinger case and the diffraction measure in the diffraction case. The key results above demonstrate that the pure point nature reflects disorder in the Schr\"odinger case and order in the diffraction case.

To contrast this, note that Schr\"odinger spectral measures are purely absolutely continuous in the periodic case, whereas the Bartlett diffraction measure (which subtracts off the trivial point mass at zero) is purely absolutely continuous in the random case.

This can be interpolated! Indeed, more generally and based on a multitude of existing results, it is by now well understood that Schr\"odinger spectral measures become more singular with increasing disorder, whereas diffraction measures become more regular with increasing disorder.

On the order-disorder spectrum, the extremal cases (periodic and random) are well understood in both settings. In the intermediate region, many cases are understood, while it is also true that many other cases are not yet understood. Let us mention some of the results. On the diffraction side, it is of course of great interest to exhibit or even characterize structures with pure point diffraction measure. This is a recent result by Lenz-Spindeler-Strungaru \cite{LSS}. Cases with singular continuous or absolutely continuous diffraction are discussed in the monograph \cite{TAO}. On the Schr\"odinger side, the survey \cite{D17} and the monographs \cite{DF22, DF24} contain a plethora of results, again showing that all spectral types occur, and do so in accordance with the general philosophy that increased disorder corresponds to more singular spectral measures. However, there is no analogous characterization of all structures that are ordered in the sense of having absolutely continuous spectral measures. There used to be a conjecture in this direction, the Kotani-Last conjecture (see the discussion in \cite{DF24}), but it has been disproved by A.~Avila \cite{A15}, Volberg-Yuditskii \cite{VY14} (see also Damanik-Yuditskii \cite{DY16}), and You-Zhou \cite{YZ15}. Nevertheless, there are many aperiodic Schr\"odinger operators with purely absolutely continuous spectral measures.

Given the existing understanding of the \emph{quasi-ordered} case, meaning structures that share the spectral type of periodic structures, it is of course equally interesting to understand the phenomenon of \emph{pseudorandomness}, meaning structures that share the spectral type of the random case. One is naturally interested in how far this can be pushed, that is, one seeks those structures with the least amount of disorder that are still pseudorandom in this sense.

To that end, let us recall the following conjecture from Bourgain \cite[Section~XV]{B05} (see also Spencer \cite{S90} and Han-Lemm-Schlag \cite{HLS20a, HLS20b}). Consider the Schr\"odinger operator
\begin{equation}\label{e.SO}
[H \psi](n) = \psi(n+1) + \psi(n-1) + V(n) \psi(n)
\end{equation}
in $\ell^2(\ZZ)$, where the potential $V : \ZZ \to \RR$ is given by
\begin{equation}\label{e.potential}
V(n) = \lambda \cos \Big( 2 \pi \Big( x_1 + n x_2 + \frac{n(n-1)}{2} \alpha \Big) \Big)
\end{equation}
with $\lambda \in \RR \setminus \{0\}$, $x_1, x_2 \in \TT = \RR / \ZZ$ and $\alpha$ irrational. The argument of the cosine is such that it is generated by iterating the \emph{skew-shift} $(x_1,x_2) \mapsto (x_1 + x_2, x_2 + \alpha)$ on $\TT^2$ and projection on the first coordinate. Indeed, for any $n \in \ZZ$, the skew-shift maps
$$
\Big( x_1 + n x_2 + \frac{n(n-1)}{2} \alpha , x_2 + n \alpha \Big) \mapsto \Big( x_1 + (n+1) x_2 + \frac{(n+1)n}{2} \alpha , x_2 + (n+1) \alpha \Big).
$$

\begin{conj}\label{conjecture}
For $\lambda \not=0$, $\alpha$ satisfying a typical Diophantine condition, and Lebesgue almost all $(x_1,x_2) \in \TT^2$, all spectral measures of the Schr\"odinger operator $H$ defined by \eqref{e.SO}--\eqref{e.potential} are pure point and all eigenfunctions decay exponentially.
\end{conj}

This conjecture is still wide open and it appears to be very difficult. Bourgain-Goldstein-Schlag \cite{BGS01} obtained partial results, which however do not fully capture the pseudorandom aspects of the model since their proofs extend arguments developed in the quasi-periodic case, which is known not to be pseudorandom. The issue is that the underlying dynamics given by the skew-shift displays too little randomness or complexity for the known methods to produce such a result to be applicable. As a result, Bourgain-Goldstein-Schlag \cite{BGS01} work in the large coupling regime, which is where similar results were established earlier in the quasi-periodic case. Further partial results, which established phenomena in line with the conjecture that do not hold in the quasi-periodic case, were obtained by Bourgain in the small coupling regime: the presence of point spectrum \cite{B02} and the positivity of the Lyapunov exponent (indicating the exponential decay of eigenvectors) \cite{B00} for many parameters. Related results in the same spirit were obtained by Kr\"uger \cite{K11, K12}.

The state of affairs makes the conjecture quite intriguing, as in this example, the apparent small amount of randomness/complexity is supposed to be sufficient to reproduce spectral properties generated by genuine random potentials. Indeed, in addition to the spectral type, the conjectures for this model described in \cite[Section~XV]{B05} also include the positivity of the Lyapunov exponent and the absence of gaps in the spectrum of $H$ -- properties shared by the Anderson model (at sufficiently small coupling).

We also want to point out that the conjecture is closely related, certainly in spirit and potentially on a technical level as well, to the work of Rudnick-Sarnak-Zaharescu \cite{RSZ01} on the distribution of spacings between the fractional parts of $n^2 \alpha$.

As the Schr\"odinger version of the conjecture is at present out of reach, it is a natural question whether the diffraction version of it can be obtained. That is, can it be shown that the diffraction of \eqref{e.potential} is pseudorandom in the sense that it is of the same type as in the genuine random case, namely absolutely continuous? We will give an affirmative answer in this paper:

\begin{theorem}\label{t.mainapplication}
For every $\lambda \not= 0$, every irrational $\alpha$, and every $x_1,x_2$, the diffraction measure associated with
$$
V(n) = \lambda \cos \Big( 2 \pi \Big( x_1 + n x_2 + \frac{n(n-1)}{2} \alpha \Big) \Big)
$$
is purely absolutely continuous.
\end{theorem}

This will arise as a special case of Theorem~\ref{t.main} below; see Corollary~\ref{cor-diff}. Our result may be viewed as a confirmation of a pseudorandomness property of this structure as well as further evidence in support of Conjecture~\ref{conjecture}. This result adds a new class to the short list of deterministic systems with absolutely continuous diffraction.

The following table summarizes results that emphasize how the seemingly innocuous change of passing from irrational circle rotations to the skew-shift on $\TT^2$ introduces pseudorandom phenomena:
\medskip
\begin{center}
\begin{tabular}{|c|c|c|}
\hline
& Schr\"odinger & diffraction \\ \hline
$V(n) = \cos(2 \pi (x + n \alpha)$ & a.c. & p.p. \\ \hline
$V(n) = \cos \big( 2 \pi \big( x_1 + n x_2 + \frac{n(n-1)}{2} \alpha \big) \big)$ & p.p. (conjectured) & a.c. \\ \hline
$V(n)$ i.i.d.\ random & p.p. & a.c. \\ \hline
\end{tabular}
\end{center}
\medskip
On a broader level, the present paper provides additional motivation for seeking a tangible way to express the apparent duality between the spectral types in the Schr\"odinger and diffraction settings.

\section{Preliminaries}

In this section we briefly review the main concepts we need for the paper.

\subsection{Uniform Distribution and Well Distribution}

Let us recall the following definition, where as usual we denote $\TT = \RR / \ZZ$.

\begin{definition} Let $x : \ZZ \to \CC$ be a bi-infinite sequence.
We say that $(x_n)_{n \in \ZZ}$ is \emph{uniformly distributed$\pmod{1}$} if for each continuous function $f: \TT \to \CC$ we have
\[
\lim_{m \to \infty} \frac{1}{2m+1} \sum_{j=-m}^m f( x_j )=\int_\TT f(x) \, \dd x \,.
\]
We say that $(x_n)_{n \in \ZZ}$ is \emph{well distributed$\pmod{1}$} if for each continuous function $f: [0,1] \to \CC$ we have
\[
\lim_{m \to \infty} \frac{1}{2m+1} \sum_{j=l-m}^{l+m} f(  x_j )=\int_\TT f(x) \, \dd x \,,
\]
uniformly in $l \in \ZZ$.
\end{definition}

It follows immediately from the definition that if $(x_n)_{n \in \ZZ}$ is uniformly distributed$\pmod{1}$, then
$$
\lim_{m \to \infty} \frac{1}{2m+1} \sum_{j=-m}^m e^{2 \pi i x_j} = \lim_{m \to \infty} \frac{1}{2m+1} \sum_{j=-m}^m \cos(2 \pi  x_j) = 0.
$$
Moreover, if $(x_n)_{n \in \ZZ}$ is well distributed$\pmod{1}$, then
$$
\lim_{m \to \infty} \frac{1}{2m+1} \sum_{j=l-m}^{l+m} e^{2 \pi i x_j} = \lim_{m \to \infty} \frac{1}{2m+1} \sum_{j=l-m}^{l+m} \cos(2 \pi  x_j) = 0,
$$
uniformly in $l$.

The concept of well distributed (bi)-sequences is closely related to the concept of amenable functions. Recall here \cite{Ebe} \cite[Def.~4.5.5]{MoSt} that $g: \ZZ \to \CC$ is called \emph{amenable} if the limit
\[
M(g):= \lim_{m \to \infty} \frac{1}{2m+1} \sum_{j=l-m}^{l+m} g(j)
\]
exists uniformly in $l$. The limit $M(g)$ then is called the (\emph{uniform}) \emph{mean of} $g$. It follows that $(x_n)_{n \in \ZZ}$ is well distributed$\pmod{1}$ if for each continuous function $f: \TT \to \CC$, the function $g(n)=f(x_n)$ is amenable with mean $M(g)=\int_\TT f(x) \, \dd x$. Recall also \cite[Remark 4.5.7]{MoSt} that a function $g$ is amenable if and only if its mean exists uniformly along one (and hence all) F\o lner sequence(s).

The following result is an immediate consequence of \cite[Theorem~2]{Law}. Note that while \cite[Theorem~2]{Law} proves well distribution, it only shows it for $(x_n)_{n \in \NN}$, while we need it for bi-infinite sequences. We leave working out the necessary changes to the proof to the reader.

\begin{theorem}\label{Thm-law} Let $P(X)=a_nX^n+\ldots +a_1X+a_0 \in \RR[X]$ be so that at least one of $a_n, \ldots, a_1$ is irrational. Then $(P(n))_{n \in \ZZ}$ is well distributed$\pmod{1}$.  In particular,
$$
\lim_{m \to \infty} \frac{1}{2m+1} \sum_{j=l-m}^{l+m} e^{2 \pi i P(j)} = \lim_{m \to \infty} \frac{1}{2m+1} \sum_{j=l-m}^{l+m} \cos(2 \pi P(j)) = 0,
$$
uniformly in $l$.
\end{theorem}

\subsection{Autocorrelation and Diffraction}

Here we recall the concepts of the autocorrelation and diffraction of bounded functions on $\ZZ$. For a brief review of the autocorrelation and diffraction measure, as well as for the details we skip below, we refer the reader to Appendix~\ref{App:diff}, while for a more general review of diffraction, we recommend \cite{TAO}.

Let us start with the following result.

\begin{proposition}\label{prop-acdef} Let $V :\ZZ \to \CC$ be bounded and let $(F_m)_m$ be a F\o lner sequence. Then, there exists a subsequence $(F_{k_m})_m$ of $(F_m)_m$ such that, for all $k \in \ZZ$, the following limit exists,
\[
\gamma_{V}(k) := \lim_{m \to \infty} \ \frac{1}{\card(F_{k_m})} \sum_{j \in F_{k_m}} V(j)\cdot \overline{V(j-k)} \,.
\]
Moreover, $\gamma_{V}$ is positive definite function on $\ZZ$ and hence there exists a finite positive measure $\sigma$ on $\TT$ such that for all $k \in \ZZ$, we have
\begin{equation}\label{eq-diff}
\gamma_V(k)= \int_{\TT} e^{2 \pi i \, k \cdot y} \, \dd \sigma (y).
\end{equation}
\end{proposition}
\begin{proof}
This follows from Lemma~\ref{lem-ac-exists} and Lemma~\ref{lem-ac-alt}.
\end{proof}

We can now introduce the following definition.

\begin{definition}\label{def-diff}
Let $V \in \ell^\infty(\ZZ)$ and let $(F_m)_{m}$ be a F\o lner sequence. We say that the \emph{autocorrelation $\gamma_V$ of $V$ exists along the sequence $(F_m)_{m}$} if for each $k \in \ZZ$, the following limit exists,
\[
\gamma_{V}(k) := \lim_{m \to \infty} \ \frac{1}{\card(F_{k_m})} \sum_{j \in F_{k_m}} V(j)\cdot \overline{V(j-k)}.
\]
In this case, the positive measure  $\sigma$ on $\TT$ satisfying Eqn.\eqref{eq-diff} is called the \emph{diffraction measure of $V$ along the sequence $(F_m)_{m}$}.

The positive measure
\[
\Gamma := \sigma- \left( \sigma(\{0\}) \delta_{0} \right)
\]
is called the \emph{Bartlett diffraction measure of $V$ along the sequence $(F_m)_{m}$}.
\end{definition}

\begin{remark}\label{rem-diff}
(a) When working over $\RR^d$, or more generally over a second countable locally compact Abelian group $G$, the autocorrelation $\gamma$ is positive definite and hence it has a positive measure $\widehat{\gamma}$ as a Fourier transform (see for example \cite[Theorem 4.11.5]{MoSt}). The positive measure $\widehat{\gamma}$ is called the \emph{diffraction measure}, and when $G=\ZZ$, $\widehat{\gamma}$ is exactly the measure $\sigma$ we defined above.

The \emph{Bartlett diffraction measure} is usually defined as
\[
\Gamma := \widehat{\gamma}- \left(\widehat{\gamma}(\{0\}) \delta_0 \right).
\]
It is important in the study of diffraction of Delone sets, as in this case there is always a trivial Bragg peak in the origin, which the measure $\Gamma$ ignores.
\\[1mm]
(b) When $\mu \, = \, \sum_{n \in \ZZ} V(n) \delta_n$ is viewed as a measure on $\RR$ (supported inside $\ZZ$), its diffraction $\widehat{\gamma}$ is a positive measure on $\RR$ which is $\ZZ$-periodic. The measure $\widehat{\gamma}$ is the only $\ZZ$-periodic measure on $\RR$ whose push-forward through the factor mapping $\pi :\RR \to \TT$ is $\sigma$.
\end{remark}

The following characterization of bounded functions with unique autocorrelation is an immediate consequence of Lemma~\ref{lem-amenab}.

\begin{proposition}\label{cor:un-ac} Let $V \in \ell^\infty(\ZZ)$. Then, the following are equivalent:
\begin{itemize}
    \item[(i)] The autocorrelation $\gamma$ of $V$ exists with respect to all F\o lner sequences.
    \item[(ii)] The autocorrelation $\gamma$ of $V$ exists and is the same with respect to all F\o lner sequences.
    \item[(iii)] For all $k \in \ZZ$, the function $j \mapsto V(j) \cdot \overline{V(j-k)}$ is amenable.
\end{itemize}
Moreover, in this case we have $\gamma(k) = M \bigl(V(j) \cdot \overline{V(j-k)}\bigr)$ for all $k \in \ZZ$.
\end{proposition}

\begin{definition}
In the situation of Proposition~\ref{cor:un-ac} we say that $V$ \emph{has a unique autocorrelation}.
\end{definition}

\subsection{The Hull of a Bounded Function}

We denote by $T: \ell^\infty(\ZZ) \to \ell^\infty(\ZZ)$ the shift operator $T V(n) := V(n+1)$. For each $V \in \ell^\infty(\ZZ)$, it is immediate that the orbit $O(V)=\{ T^n V : n \in \ZZ\}$ has compact closure $\XX(V)$ in the topology of pointwise converegence. Moreover, the topology of pointwise convergence is metrizable on $\XX(V)$. Therefore $\XX(V)$ is a topological dynamical system. If $\XX(V)$ is uniquely ergodic, we will simply say that \emph{$V$ is uniquely ergodic}.

For each $n \in \ZZ$, we can define a continuous function $E_n : \XX(V) \to \CC$ via $E_n(w)=w(n)$. Note that for all $n \in \ZZ$, we have $E_0(T^n w)=w(n)$.

Next, let us recall that any shift invariant measure $\mm$ on $\XX(V)$ induces a unitary representation of $\ZZ$ onto $L^2(\XX(V), \mm)$ via the shift action $U_n(f) := T^{n} f$. Thus, via the Stone representation theorem, to each $f \in L^2(\XX(V), \mm)$ we can associate a positive measure $\sigma_f$ on $\TT$ such that for all $k \in \ZZ$,  we have
\[
\langle  f, T^{k} f \rangle \, = \, \int_{\TT} e^{2 \pi i \, k \cdot y} \, \dd \sigma_f(y).
\]
The measure $\sigma_f$ is called the \emph{spectral measure of $f$}.

Let us now recall the following result (\cite{Dwo,LMS,BL,Gou}), usually called the Dworkin Argument,  relating the diffraction and spectral measures for uniquely ergodic dynamical systems. More general results in this direction, which hold for systems that are not uniquely ergodic, can be found in \cite{BL} and \cite{LSS,LS}, but the uniquely ergodic case is the one relevant for us. Since the proof is straightforward for shift actions, we include it here for completeness.

\begin{proposition}[Dworkin Argument]\label{prop-Dwo} Let $V \in \ell^\infty(\ZZ)$ be uniquely ergodic with unique ergodic measure $\mm$. Then, for all $w \in \XX(V)$,
the autocorrelation $\gamma$ of $V$ exists along all F\o lner sequences and satisfies $\gamma(k)= \langle  E_0, T^{k} E_0 \rangle$ for all $k \in \ZZ$. In particular, each $w \in \XX(V)$ has a unique autocorrelation, and its diffraction measure $\sigma$ satisfies $\sigma=\sigma_{E_0}$.
\end{proposition}

\begin{proof}
Let $k \in \ZZ$ and let $(F_m)_{m}$ be any F\o lner sequence.
Then, since $E_0 \cdot \overline{T^k E_0} \in C(\XX(V))$, by unique ergodicity  we have
\begin{align*}
 \langle E_0, T^{k} E_0 \rangle &=
% \int_{\XX(V)} E_0(\varpi) \cdot  \overline{T^{k} E_0(\varpi)} \, \dd \mm(\varpi)  \\ & =
 \lim_{m \to \infty} \frac{1}{\card(F_m)} \sum_{j \in F_m} E_0(T^j w)  \cdot \overline{T^k E_0(T^j w) } \\
% & = \lim_{m \to \infty} \frac{1}{\card(F_m)} \sum_{j \in F_m} E_0(T^j w) \cdot \overline{ E_0(T^{j-k} w) } \\
 & = \lim_{m \to \infty} \frac{1}{\card(F_m)} \sum_{j \in F_m}  w(j) \cdot   \overline{ w (j-k) }.
\end{align*}
The claim follows.
\end{proof}

We complete this section by giving the following characterization for the eigenvalues of $\XX(V)$.

\begin{proposition}\label{prop:top-eigen}
Assume that $V \in \ell^\infty(\ZZ)$ is real valued and uniquely ergodic. Let $\beta \in \TT = \RR / \ZZ$ and assume that for all not necessarily distinct $n_1, \ldots, n_r \in \ZZ$, the function
  \[
  \ZZ \ni j \to \left( \prod_{k=1}^r V(j+n_k) \right) \cdot e^{-2 \pi i  \beta  j}
  \]
is amenable. Then, the following are equivalent:
\begin{itemize}
  \item[(i)] $\beta$ is an eigenvalue.
  \item[(ii)]$\beta$ is a topological eigenvalue.
  \item[(iii)] There exist not necessarily distinct $n_1, \ldots, n_r \in \ZZ$ such that
  \[
  \lim_{m \to \infty} \frac{1}{2m+1} \sum_{j=-m}^m \left( \prod_{k=1}^r V(j+n_k) \right) \cdot e^{-2 \pi i  \beta j} \neq 0 \,.
  \]
\end{itemize}
\end{proposition}
\begin{proof}
Since $V$ is real valued, we have for all $m \in \ZZ$, $\overline{V(j+m)} =  V(j+m)$. The claim then follows from Proposition~\ref{propC7}.
\end{proof}

\subsection{$W$-polynomials}

In this section we discuss some basic properties of $W$-polynomials, defined below.  These properties will allow us to calculate the autocorrelation and diffraction of
$V(n) = \lambda \cos(2 \pi P(n))$.

\begin{definition}\label{def-w-poly}
Let $P(X)=a_nX^n + \ldots +a_1X+a_0 \in \RR[X]$ and $d \geq 1$. We say that $P(X)$ is a \emph{$W_d$-polynomial} if $d= \max \{ j : 1 \leq j \leq n , a_j \notin \QQ \}$. We will simply say that $P(X)$ is a \emph{$W$-polynomial} if $P(x)$ is a $W_d$-polynomial for some $d \geq 1$.
\end{definition}

Let us note here in passing that Theorem~\ref{Thm-law} and Proposition~\ref{prop1} below imply that a polynomial $P(X) \in \RR[X]$ is a $W$-polynomial if and only if $(P(n))_{n \in \ZZ}$ is well distributed$\pmod{1}$, which explains the terminology.

When calculating the autocorrelation, given a polynomial $P$, we will need to calculate the mean of $\cos(2 \pi S_k(P)(j))$ and $\cos(2 \pi D_k(P)(j))$, where
\begin{align*}
S_k(P)(X)&:= P(X)+P(X+k)  \qquad \forall k \in \ZZ  \, \mbox{ and } \\
D_k(P)(X)&:= P(X+k)-P(X)  \qquad \forall 0 \neq k \in \ZZ  \,.
\end{align*}
For $W$-polynomials, these means are $0$, so a natural question to ask is the following:

\begin{quest}
\begin{itemize}
  \item[(Q1)] For which polynomials $P(X) \in \RR[X]$ are all polynomials $S_k(P)$, $k \in \ZZ$, $W$-polynomials?
  \item[(Q2)] For which polynomials $P(X) \in \RR[X]$ are all polynomials $D_k(P)$, $0 \neq k \in \ZZ$, $W$-polynomials?
\end{itemize}
\end{quest}

The key observation to answer these questions is that if a polynomial $P(X)$ satisfies the condition from (Q1) or (Q2), respectively, then, adding any polynomial in $\QQ[X]$ to $P(X)$ yields another polynomial that satisfies the condition from (Q1) or (Q2), respectively.

This allows us to ignore all monomials in $P$ with rational coefficients. To make this precise, we introduce the following notation: For $P(X)=\sum_{j=0}^n a_j X^j \in \RR[X]$, let
$$
P_{\mathsf{rat}}(X) := \sum_{\substack{j=0 \\ a_j \in \QQ}}^n a_jX^j, \quad P_{\mathsf{irr}}(X) := \sum_{\substack{j=0 \\ a_j \notin \QQ}}^n a_jX^j
$$
denote the sums of monomials with rational and irrational coefficients, respectively.

\begin{lemma}\label{lem-aux}
\begin{itemize}
  \item[(a)] Let $P(X) \in \RR[X]$ and $d \geq 1$. Then, $P(X)$ is a $W_d$-polynomial if and only if $\deg(  P_{\mathsf{irr}}(X)) = d$.
  \item[(b)] For all $P(X) \in \RR(X)$ and all $k \in \ZZ$, we have
\begin{align*}
  S_k(P) &= S_k(P_{\mathsf{rat}})+ S_k(P_{\mathsf{irr}}) \,, \\
  D_k(P) &= D_k(P_{\mathsf{rat}})+ D_k(P_{\mathsf{irr}}) \,, \\
  S_k(P_{\mathsf{rat}})&;  D_k(P_{\mathsf{rat}})   \in \QQ[X] \,.
\end{align*}
\item[(c)] For all $P(X) \in \RR(X)$ and all $k \in \ZZ$, we have
\begin{align*}
  \deg(S_k(P)) &= \deg(P) \,, \\
  \mbox{leading coefficient of } S_k(P) &= 2 \cdot \mbox{leading coefficient of } P \,.
\end{align*}
\item[(d)] For all $P(X) \in \RR(X)$ with $\deg(P) \geq 1$ and all $0 \neq k \in \ZZ$, we have
\begin{align*}
  \deg(D_k(P)) &= \deg(P)-1 \,,\\
  \mbox{ leading coefficient of } D_k(P) &=k \cdot \deg(P) \cdot \mbox{leading coefficient of } P \,.
\end{align*}
\end{itemize}
\end{lemma}
\begin{proof}
\textbf{(a), (b), (c)} are obvious.

\noindent \textbf{(d)} Let $P(X)=a_nX^n+\ldots+a_1X+a_0 \in \RR[X]$ with $a_n \neq 0$. Then,
\begin{align*}
 D_k(P) &=\left(a_{n}(X+k)^n+a_{n-1}(X+k)^{n-1}+ \ldots + a_1(X+k)+a_0\right) \\
 &- \left(a_{n}X^n+a_{n-1}X^{n-1}+ \ldots + a_1X+a_0\right) \,.
\end{align*}
It follows that $D_k(P)$ has degree at most $n$. The only two terms containing $X^n$ are $a_{n}X^n$ in the first bracket, and $-a_{n}X^n$ coming from the second bracket. Those cancel out. Next, we get the following three terms containing $X^{n-1}$: $a_nnkX^{n-1}+a_{n-1}X^{n-1}- a_{n-1}X^{n-1}=a_nnkX^{n-1}$. It follows that the coefficient of $X^{n-1}$ is $a_n \cdot n \cdot k$. Since $a_n, n, k$ are all non-zero, the claim follows.
\end{proof}

We can now answer (Q1):

\begin{lemma} Let $P \in \RR(X)$ and $d \geq 1$. Then, the following are equivalent:
\begin{itemize}
  \item[(i)] $P$ is a $W_d$-polynomial.
  \item[(ii)] $\deg(P_{\mathsf{irr}}) =d$.
  \item[(iii)] $S_k(P)$ is a $W_d$-polynomial for all $k \in \ZZ$.
  \item[(iv)] There exists some $k \in \ZZ$ so that $S_k(P)$ is a $W_d$-polynomial.
\end{itemize}
\end{lemma}
\begin{proof}
The equivalence {\bf (i) $\Leftrightarrow$ (ii)} follows from Lemma~\ref{lem-aux}.

\smallskip
\noindent{\bf (ii) $\Longrightarrow$ (iii):} Let $m \in \ZZ$ be fixed but arbitrary. By Lemma~\ref{lem-aux}(c), $S_m(P_{\mathsf{irr}})$ is a $W_d$-polynomial. As $S_m(P_{\mathsf{rat}}) \in \QQ[X]$ we get that $S_m(P)$ is a $W_d$-polynomial.

\smallskip
\noindent{\bf (iii) $\Longrightarrow$ (iv):} is obvious.

\medskip
\noindent {\bf (iv) $\Longrightarrow$ (ii):}  Since $S_k(P)$ is a $W_d$-polynomial, and $S_k(P_{\mathsf{rat}}) \in \QQ(X)$ we get that $S_k(P_{\mathsf{irr}})= S_k(P)-S_k(P_{\mathsf{rat}})$ is a $W_d$-polynomial. Thus, Lemma~\ref{lem-aux} (c) implies that $\deg(P_{\mathsf{irr}}) =d$.
\end{proof}
This immediately yields:

\begin{corollary}\label{cor-SkWpol}  Let $P \in \RR(X)$. Then, the following are equivalent:
\begin{itemize}
  \item[(i)] $P$ is a $W$-polynomial.
  \item[(ii)] $\deg(P_{\mathsf{irr}}) \geq 1$.
  \item[(iii)] $S_k(P)$ is a $W$-polynomial for all $k \in \ZZ$.
  \item[(iv)] There exists some $ k \in \ZZ$ so that $S_k(P)$ is a $W$-polynomial.
\end{itemize}
\end{corollary}

Now, we give the answer to (Q2).

\begin{lemma} Let $P(X) \in \RR(X)$ and let $d \geq 2$.
Then, the following are equivalent:
\begin{itemize}
  \item[(i)] $P$ is a $W_d$-polynomial.
  \item[(ii)] $\deg(P_{\mathsf{irr}}) =d$.
  \item[(iii)] $D_k(P)$ is a $W_{d-1}$-polynomial for all $0 \neq k \in \ZZ$.
  \item[(iv)] There exists some $0\neq k \in \ZZ$ so that $D_k(P)$ is a $W_{d-1}$-polynomial.
\end{itemize}
\end{lemma}
\begin{proof}
The equivalence {\bf (i) $\Leftrightarrow$ (ii)} has already been established.

\smallskip
{\bf (ii) $\Longrightarrow$ (iii):} Let $m \in \ZZ$ be fixed but arbitrary.

By Lemma~\ref{lem-aux}(d), $D_m(P_{\mathsf{irr}})$ is a $W_{d-1}$-polynomial. As $D_m(P_{\mathsf{rat}}) \in \QQ[X]$ we get that $D_m(P)$ is a $W_{d-1}$-polynomial.

\smallskip
{\bf (iii) $\Longrightarrow$ (iv):} is obvious.

\medskip
{\bf (iv) $\Longrightarrow$ (ii):}  Since $D_k(P)$ is a $W_{d-1}$-polynomial, and $D_k(P_{\mathsf{rat}}) \in \QQ(X)$ we get that $D_k(P_{\mathsf{irr}})= D_k(P)-D_k(P_{\mathsf{rat}})$ is a $W_{d-1}$-polynomial. Thus, Lemma~\ref{lem-aux} (d) implies that $\deg(P_{\mathsf{irr}}) =d$.
\end{proof}

\begin{corollary}\label{cor-DkWpol} Let $P \in \RR(X)$. Then, the following are equivalent:
\begin{itemize}
  \item[(i)] $P$ is a $W_d$-polynomial for some $d \geq 2$.
  \item[(ii)] $\deg(P_{\mathsf{irr}}) \geq 2$.
  \item[(iii)] $D_k(P)$ is a $W$-polynomial for all $0\neq k \in \ZZ$.
  \item[(iv)] There exists some $0 \neq m \in \ZZ$ so that $D_k(P)$ is a $W$-polynomial.
\end{itemize}
\end{corollary}

% Combining all results in this section we get the following:

% \begin{theorem}\label{T1} Let $P(X) =a_dX^d+\ldots +a_1X+a_0 \in \RR[X]$ be a $W_d$-polynomial for some $d \geq 2$. Then,
% %
% \begin{align*}
%   \lim_{m \to \infty} \frac{1}{2m+1} \sum_{j=l-m}^{l+m} \cos( 2\pi P(j)) \cos(2 \pi P(j+k)) &=0  \qquad \forall k \neq 0  \\
%  \lim_{m \to \infty} \frac{1}{2m+1} \sum_{j=l-m}^{l+m} \cos( 2\pi P(j)) \cos(2 \pi P(j)) &= \frac{1}{2} \,,
% \end{align*}
% %
% uniformly in $l$.
% \end{theorem}
% \begin{proof} For all $k \in \ZZ$ we have
% %
% \[
% \cos( 2\pi P(j)) \cos(2 \pi P(j+k)) = \frac{1}{2} \left( \cos(2\pi S_k(P)(j)) + \cos(2\pi D_k(P)(j)) \right) \,.
% \]
% %
% Now, the previous results together with $D_0(P)=0$ and Theorem~\ref{Thm-law}   proves the claim.
% \end{proof}

\subsection{Non $W$-polynomials.}

Let $P(X) = a_nX^n + \ldots +a_1X+a_0 \in \RR[X]$ be so that $a_n, \ldots, a_1 \in \QQ$. Then, setting $N$ to be the common denominator of $a_n, \ldots, a_1$, we can write $P(X)=\frac{Q(X)}{N}+c$ for some $Q \in \ZZ[X]$ and $c \in \RR$. Now, recalling that for $a,b \in \ZZ$ and $Q \in \ZZ[X]$, we have $(b-a)|(Q(b)-Q(a))$, we get that $P(j+N)-P(j) \in \ZZ$ for every $j \in \ZZ$. Therefore, we obtain

\begin{proposition}\label{prop1}
Let $P(X)=a_nX^n+\ldots+a_1X+a_0 \in \RR[X]$ be so that $a_n, \ldots, a_1 \in \QQ$, and let $N$ to be the common denominator of $a_n, \ldots, a_1$.
Then, $f(j) := e^{2 \pi i P(j)}$ and $g(j) := \cos(2 \pi  P(j))$ are periodic with period $N$. In particular, $f,g$ are amenable and
$$
M(f) = \frac{1}{N} \sum_{j=1}^N e^{2 \pi i P(j)}, \quad M(g) = \frac{1}{N}\sum_{j=1}^N \cos(2 \pi i P(j)).
$$
\end{proposition}

The following is an easy exercise and we skip the proof.

\begin{lemma}\label{lem1}
Let $h : \ZZ \to \CC$ be $N$ periodic. Then there exist $c_1, \ldots, c_N \in \CC$ such that
\[
h(j) = \sum_{l=1}^N c_l e^{2 \pi i \frac{j \cdot l}{N}},
\]
that is, $h$ agrees with a trigonometric polynomial.
\end{lemma}

\begin{remark}\label{rem-W1} Let $P(X) = a_nX^n + \ldots +a_1X+a_0 \in \RR[X]$ be a $W_1$-polynomial. Then $P(X)-a_1X$ is not a $W$-polynomial. Therefore, by Proposition~\ref{prop1} and Lemma~\ref{lem1} there exist some $N$ and $c_1, \ldots, c_N \in \CC$ such that
\begin{align*}
 e^{2 \pi i P(j)}&=  e^{2 \pi i P(j)-a_1 j} e^{2 \pi i a_1 j}=  \sum_{j=1}^N c_j e^{2 \pi i \left(\frac{l}{N}+a_1 \right)j}\\
 \cos(2 \pi i P(j))&=  \frac{1}{2}   \sum_{j=1}^N c_j e^{2 \pi i \left(\frac{l}{N}+a_1 \right)j}+ \frac{1}{2}e^{-2 \pi i \left(\frac{l}{N}+a_1 \right)j} \,.
\end{align*}
\end{remark}

\section{Diffraction of $V(n)  = a+b  \cos(2 \pi P(n))$ }

This paper deals with the diffraction measure and the dynamical system  of the pseudorandom function $V : \ZZ \to \RR$ given by
$$
V(n) = a + b \cos \Big( 2 \pi \Big( x_1 + n x_2 + \frac{n(n-1)}{2} \alpha \Big) \Big),
$$
where $a, b \in \CC$, $x_, x_2 \in \RR$, and $\alpha \in \RR \backslash \QQ$. In fact, we can study more generally the diffraction of $V(n) :=  a+b\cos( 2 \pi P(n))$ for a polynomial $P$.

\smallskip

Let us start with the following preliminary result, which gives the interesting case Theorem~\ref{t.main}(a) below, and can be seen as a generalisation of it.

\begin{lemma}\label{diff-cos-wd}
Let $(y_n)_{n \in \ZZ}$ be a sequence such that
\begin{itemize}
\item{} The sequence $(y_n)_{n \in \ZZ}$ is well distributed.
\item{} For each $k \in \ZZ$ the sequence $(y_n+y_{n+k})_{n \in \ZZ}$ is well distributed.
\item{} For each $0\neq k \in \ZZ$ the sequence $(y_n-y_{n+k})_{n \in \ZZ}$ is well distributed.
\end{itemize}
Let $a,b \in \CC$ and let
\[
V(n) \, := \, a+b \cos(2 \pi y_n) \,.
\]
Then, $V$ has a unique autocorrelation and
$$
\gamma = |a|^2+\frac{|b|^2}{2} \delta_0, \quad \sigma = |a|^2\delta_{0+\ZZ}+\frac{|b|^2}{2} \theta_{\TT},
$$
where $\theta_{\TT}$ is the Haar probability measure on $\TT$.
\end{lemma}

\begin{proof} Let $j,k \in \ZZ$. Then
\begin{align*}
V(j) \cdot \overline{V(j-k)}&= \left( a+b \cos(2 \pi y_{j}) \right) \cdot \left( \overline{a}+\overline{b} \cos(2 \pi y_{j-k}) \right) \\
&= |a|^2+a \cdot \bar{b} \cos(2 \pi y_{j-k}) +\bar{a} \cdot b \cos(2 \pi y_{j})+ |b|^2 \cos(2 \pi y_{j})\cos(2 \pi y_{j-k})  \\
&= |a|^2+a \cdot \bar{b} \cos(2 \pi y_{j-k}) +\bar{a} \cdot b \cos(2 \pi y_{j})\\
&+ \frac{|b|^2}{2}  \cos(2 \pi (y_{j}+ y_{j-k})+ \frac{|b|^2}{2}  \cos(2 \pi (y_{j}- y_{j-k}).
\end{align*}
Now, for each fixed $k \neq 0$, by the well-distribution of $y_j, y_{j-k}, y_{j}+y_{j-k}$ and $y_{j}-y_{j-k}$, the following mean exists uniformly and obeys $M(V(j) \cdot \overline{V(j-k)})=|a|^2$. On the other hand, when $k=0$, the sequences $y_j, y_{j-k}, y_{j}+y_{j-k}$ are still well distributed, but $y_{j}-y_{j-k}=0$. Therefore, in this case,  the following mean exists uniformly and obeys $M(V(j) \cdot \overline{V(j-k)})=|a|^2 +\frac{|b|^2}{2}$. The claim follows.
\end{proof}

\begin{remark}
If we replace ``well distributed'' by ``uniformly distributed'' in Lemma~\ref{diff-cos-wd}, then we can conclude that the autocorrelation and diffraction exist with respect to $F_m=\{ -m, \ldots, m\}$ and are given by the same formulas.
\end{remark}

\begin{theorem}\label{t.main}
Let $P(X)=a_dX^d+\ldots+a_1X+a_0 \in \RR[X]$, and consider for $a,b \in \CC$,
\[
V(n) \, : =\, a+b  \cos(2 \pi P(n))  \,.
\]
Then, the autocorrelation $\gamma$ of $V$ exists and is the same with respect to all F\o lner sequences. Moreover,
\begin{itemize}
  \item[(a)] If there exists some $j \geq 2$ such that $a_j \notin \QQ$, then
$$
\gamma = |a|^2+\frac{|b|^2}{2} \delta_0, \quad \sigma = |a|^2\delta_{0+\ZZ}+\frac{|b|^2}{2} \theta_{\TT} \,.
$$
  \item[(b)] If $a_n,\ldots, a_2 \in \QQ$ and $a_1 \notin \QQ$, then $V$ is a trigonometric polynomial on $\ZZ$ and there exists some $N$ such that
\[
\supp(\sigma) \subseteq  \Big\{ 0 , \pm a_1 , \pm a_1+ \frac{1}{N}, \ldots , \pm a_1+ \frac{N-1}{N} \Big\} .
\]
    \item[(c)] If $a_n \ldots, a_1 \in \QQ$, then there exists some $N$ such that $V$ is a $N$-periodic function on $\ZZ$ and
    \[
  \supp(\sigma) \, \subseteq \, \Big\{ 0, \frac{1}{N}, \ldots , \frac{N-1}{N} \Big\} .
  \]
\end{itemize}
\end{theorem}

\begin{remark}
The uniqueness of the autocorrelation also follows from Proposition~\ref{prop-Dwo} and Theorem~\ref{theorem1}(a) below.
\end{remark}

\begin{proof}
\noindent \textbf{(a)} Since $P(X)$ is a $W_j$-polynomial with $j\geq 2$, by Corollary~\ref{cor-SkWpol} and Corollary~\ref{cor-DkWpol}, $S_k(P)$ is a $W$-polynomial for all $k \in \ZZ$ and $D_k(P)$ is a $W$-polynomial for all $0 \neq k \in \ZZ$.

Therefore, by Theorem~\ref{Thm-law}, the sequence $(P(n))_{n \in \ZZ}$ is well distributed, the sequences $(P(n)+P(n+k))_{n \in \ZZ}$ are well distributed for all $k \in \ZZ$ and the sequences $(P(n)-P(n+k))_{n \in \ZZ}$ are well distributed for all $0 \neq k \in \ZZ$. The claim follows from Lemma~\ref{diff-cos-wd}.

\smallskip \noindent \textbf{(b)} Follows from Remark~\ref{rem-W1} and Corollary~\ref{cor:trig-poly}.

\smallskip \noindent   \textbf{(c)}
This follows from Proposition~\ref{prop1}, Lemma~\ref{lem1} and Corollary~\ref{cor:trig-poly}.
%Let $Q(X) \,:=\, P(X)-cX$, where $c=a_1$ in (b) and $c=0$ in (c). Then, $Q(X) \in \QQ[X]$ and $P(X)=Q(X)+c X$. Therefore,
%%
%\begin{align*}
%V(n) \, &=\,  a+b  \cos(2 \pi P(n)) = a+\frac{b}{2} e^{2 \pi i P(n)}+ \frac{b}{2} e^{-2 \pi i P(n)} \\
%\, &= \,  a+\frac{b}{2} e^{2 \pi i Q(n)} e^{2 \pi i c n}+ \frac{b}{2} e^{-2 \pi i Q(n)} e^{-2 \pi i c n} \,.
%\end{align*}
%%
%Now, by Proposition~\ref{prop1} and Lemma~\ref{lem1} there exist $N \in \NN$ and $c_1, \ldots, c_N \in \CC$ such that
%%
%\[
%e^{2 \pi i Q(n)} = \sum_{j=1}^N c_j e^{2 \pi i \frac{j}{N}} \,.
%\]
%%
%Conjugating we get
%%
%\[
%e^{-2 \pi i Q(n)} = \sum_{j=1}^N \overline{c_j} e^{-2 \pi i \frac{j}{N}} \,.
%\]
%%
%This gives
%%
%\[
%V(n)= a+\sum_{j=1}^N \frac{b}{2} c_j e^{2 \pi i (c+\frac{j}{N})}+\sum_{j=1}^N \frac{b}{2} \overline{c_j} e^{-2 \pi i (c+\frac{j}{N})} \,.
%\]
%%
%Therefore, $V$ is a trigonometric polynomial on $\ZZ$ and by Corollary~\ref{cor:trig-poly} we have
%%
%\[
%\supp(\sigma) \subseteq  \{0 , \pm c , \pm c+ \frac{1}{N}, \ldots , \pm c+ \frac{N-1}{N} \}+ \ZZ \,.
%\]
%%
%(b) and (c) now follow.
% \smallskip \noindent (c)  By Proposition~\ref{prop1} and Lemma~\ref{lem1} there exists some $N$ and $c_1, \ldots, c_N \in \CC$ such that
% %
% \[
% \cos(2 \pi P(n)) = \sum_{j=1}^N c_j e^{2 \pi i \frac{j}{N}} \,,
% \]
% %
% that is $V$ is a trigonometric polynomial on $\ZZ$ and and by Corollary~\ref{cor:trig-poly} we have
% %
% \[
% \supp(\sigma) = \{0 ,  \frac{1}{N}, \ldots , \frac{N-1}{N} \}+ \ZZ \,.
% \]
% %
\end{proof}

As immediate consequences we get:

\begin{corollary}\label{cor-diff}
Let $\lambda, x_1, x_2 \in \RR$ and $\alpha \in \RR \backslash \QQ$. Then,
\[
V_0(n) \, =\,  \lambda \cos \Big( 2 \pi \Big( x_1 + n x_2 + \frac{n(n-1)}{2} \alpha \Big) \Big)
\]
and
\[
V_1(n) \, =\,  \lambda \cos \Big( 2 \pi \alpha n^2 \Big).
\]
both have unique autocorrelation $\gamma$ and diffraction $\sigma$ given by
 \[
    \gamma =\frac{|\lambda|^2}{2} \delta_0  \qquad \mbox{ and } \qquad
    \sigma =  \frac{|\lambda|^2}{2}\theta_{\TT} \,.
\]
\end{corollary}

In particular, this establishes Theorem~\ref{t.mainapplication} from the introduction.

\section{Dynamics of $V(n)  = a+b  \cos(2 \pi P(n))$ }

In this section we study the dynamical spectrum of $V(n)  = a+b  \cos(2 \pi P(n))$.

\begin{theorem}\label{theorem1} Let $a,b \in \RR$ and let $P(X) \in \RR[X]$. Let
\[
V(n)  \, =\,  a+b  \cos(2 \pi P(n)) \,.
\]
Then,
\begin{itemize}
  \item[(a)] $V$ is uniquely ergodic.
  \item[(b)] All eigenvalues of $\XX(V)$ are topological.
  \item[(c)] If $\beta \in \TT$ is an eigenvalue for $\XX(V)$, then there exist some not necessarily distinct $m_1, \ldots, m_s$, $q_1, \ldots, q_t \in \ZZ$ such that
  \begin{align*}
  Q(X) \, &:= \, P(X+m_1)+P(X+m_2)+\ldots + P(X+m_s) \\
  &-P(X+q_1)-P(X+q_2)-\ldots - P(X+q_t)-\beta X
  \end{align*}
  satisfies
      \[
      \lim_{m \to \infty} \frac{1}{2m+1} \sum_{j=-m}^m e^{2 \pi i Q(j) } \neq 0 \,.
      \]
  In particular, all non-constant terms of $Q(X)$ must be rational. In this case, if $ P(X)=a_dX^d+\ldots +a_1X+a_0$, then $\beta \in \mbox{\rm Span}_{\ZZ}\{ a_d, \ldots, a_1 \} +\QQ$.
  \end{itemize}
\end{theorem}

\begin{remark} The same result can also be proved for $a,b \in \CC$ (by using Theorem~\ref{thm:uechar} (a) instead of Theorem~\ref{thm:uechar} (b) and Proposition~\ref{propC7} instead of Proposition~\ref{prop:top-eigen}), but since the proof is slightly more complicated and we only need the real valued case, we restrict here to the latter case.
\end{remark}

\begin{proof}
Let us first note that for all not necessarily distinct $n_1, \ldots, n_r \in \ZZ$, we have
\begin{align*}
\prod_{k=1}^r \left( V(j+n_k) -a \right)&= \prod_{k=1}^r  \left( b  \cos(2 \pi P(j+n_k)) \right) = \frac{b^r}{2^r}  \prod_{k=1}^r \left( e^{2 \pi i P(j+n_k)}+ e^{-2 \pi i P(j+n_k)}\right) \\
\, &= \,  \frac{b^r}{2^r} \sum_{e_1, \ldots, e_r \in \{ \pm 1 \}} e^{2 \pi i \sum_{k=1}^r e_k P(j+n_k)} \,.
\end{align*}
Now, for each $\ee = (e_1, \ldots, e_r) \in \{ \pm 1 \}^r =: M_r$, denote
\[
P_\ee(X) \, := \, \sum_{k=1}^r e_k P(X+n_k) \,.
\]
Then,
\begin{equation}\label{eq1}
\prod_{k=1}^r \left( V(j+n_k) -a \right)=  \frac{b^r}{2^r} \sum_{\ee \in M_r} e^{2 \pi i P_\ee(j) } \,.
\end{equation}
\noindent \textbf{(a)} Let $n_1, \ldots, n_r \in \ZZ$ be not necessarily district. Then, by \eqref{eq1}, Theorem~\ref{Thm-law} and Proposition~\ref{prop1}  the function
\[
\ZZ \ni j \to \prod_{k=1}^r \left( V(j+n_k) -a \right)
\]
is amenable. (a) follows from Theorem~\ref{thm:uechar} (b).

\smallskip \noindent \textbf{(b)}  By \eqref{eq1} we have
\[
\left(\prod_{k=1}^r \left( V(j+n_k) -a \right)\right)  e^{2 \pi i \beta  j} =  \frac{b^r}{2^r} \sum_{\ee \in M_r} e^{2 \pi i (P_\ee(j) -
\beta j)}.
\]
Since this function is amenable by Theorem~\ref{Thm-law} and Proposition~\ref{prop1}, (b) follows from Proposition~\ref{prop:top-eigen}.

\smallskip \noindent \textbf{(c)} By Proposition~\ref{prop:top-eigen}, since $\beta$ is an eigenvalue, there exist not necessarily distinct $n_1, \ldots, n_r \in \ZZ$ such that
\[
\prod_{k=1}^r \left( V(j+n_k) -a \right)  e^{2 \pi i \beta j} =   \frac{b^r}{2^r} \sum_{\ee \in M_r} e^{2 \pi i (P_\ee(j) -
\beta j)}
\]
has non-zero mean. This implies that there exists some $\ee \in M_r$ such that $e^{2 \pi i \left(P_\ee(j)-\beta j\right)}$ has non-zero mean, which proves the first claim. Theorem~\ref{Thm-law}  implies then that all non-constant terms of $P_\ee(X)-\beta X$ are rational. The last claim follows by observing that the coefficient of $X$ in $P_\ee(X)-\beta X$ is rational.
\end{proof}

Calculating the spectral measures for arbitrary polynomials turns out to be extremely technical, for this reason we restrict to the particular case where $P(X)=\alpha X^2+a_1x+a_0$ is a $W_2$-polynomial. Let us note here first that when $P(X)$ is a not a $W$-polynomial, by Proposition~\ref{prop1} $V(n)= \cos(2 \pi P(n))$ is periodic, and in this case, the hull $\XX(V)$ is a finite Abelian group. On the other hand, if $P(X)$ is a $W_1$ polynomial, then by Remark~\ref{rem-W1}, $V(n)= \cos(2 \pi P(n))$ is a trigonometric polynomial and, in this case, the hull $\XX(V)$ is a compact Abelian group. In both situations, the dynamical spectrum is pure point.

Recall first from Appendix~\ref{App:B} that for a bounded real valued $V : \ZZ \to \CC$, $\AAA$ denotes the algebra generated by $\{1_{\XX(V)} \} \cup \{ E_n : n \in \ZZ\}$; $\AAA$ is dense in $C(\XX(V))$.

\begin{theorem}\label{T2}
Let $a,b, \alpha, a_1, a_0 \in \RR$ be so that $b \neq 0$ and $\alpha \notin \QQ$. Let
\[
V(n)  \, =\,  a+b  \cos(2 \pi (\alpha n^2+ a_1 n+a_0)).
\]
 Then,
\begin{itemize}
  \item[(a)] $V$ is uniquely ergodic.
  \item[(b)] The set of eigenvalues is $2 \alpha \ZZ$ and all eigenvalues are topological..
  \item[(c)] For each $f \in \AAA$, there exists a pure point measure $\mu$ with finite support and a trigonometric polynomial $T$ such that
\[
\sigma_f = \mu+ T \theta_{\TT} \,.
\]
In particular the dynamical spectrum is mixed pure point and absolutely continuous.
\end{itemize}
\end{theorem}

\begin{remark}
With additional effort, the same result can again be proved for $a,b \in \CC$.
\end{remark}

\begin{proof}
Let $P(X):=  \alpha X^2+ a_1X +a_0$.

\smallskip \noindent \textbf{(a)} and the second statement in \textbf{(b)} follow from Theorem~\ref{theorem1}.

For the first statement in \noindent \textbf{(b)}, suppose first that $\beta$ be an eigenvalue. By Theorem~\ref{theorem1} there exist $m_1, m_2, \ldots, m_s$, $q_1, q_2, \ldots, q_t$ such that
\[
\lim_{m \to \infty} \frac{1}{2m+1} \sum_{j=-m}^m e^{2 \pi  i Q(j)} \neq 0,
\]
where
\begin{align*}
Q(X) &:=  P(X+m_1)+P(X+m_2)+\ldots + P(X+m_s)^2-P(X+q_1)^2-P(X+q_2)^2 \\
&-\ldots - P(X+q_t)^2-\beta X \,.
\end{align*}
In particular, $Q$ is a $W$-polynomial.

A short computation gives $Q(X)=AX^2+BX+C$, where
\begin{align*}
A&= \alpha (s-t)                \\
B&=  2 \alpha (m_1+m_2+\ldots +m_s - q_1-q_2-\ldots -q_t)+ a_1(s-t)
- \beta
%           \\
%C&= \alpha (m_1^2+m_2^2+\ldots +m_s^2 - q_1^2-q_2^2-\ldots -q_t^2) \\
%& \quad +   a_1 (m_1+m_2+\ldots +m_s - q_1-q_2-\ldots -q_t)   +a_0(s-t)
.
\end{align*}

Since $Q$ is a $W$-polynomial, we get that $A \in \QQ$, and hence $s-t=0$. This gives that $Q(X)=BX+C$. Since
\[
0 \neq \lim_{m \to \infty} \frac{1}{2m+1} \sum_{j=-m}^m e^{2 \pi  i Q(j)}  = e^{2 \pi  i C} \lim_{m \to \infty} \frac{1}{2m+1} \sum_{j=-m}^m e^{2 \pi  i Bj} ,
\]
the character $e^{2 \pi  i Bj}$ has non-zero mean and hence it is trivial. Thus $B=0$, and hence $\beta \in 2 \alpha \ZZ$.

Conversely, let us show that $2\alpha$ is an eigenvalue. Then, since $\mm$ is ergodic, the spectrum is a group and hence each element of $2 \alpha \ZZ$ is an eigenvalue.
%The strategy of the proof is simple: By Proposition~\ref{prop:top-eigen} we must find  not necessarily distinct $n_1, \ldots, n_r, m_1, \ldots, m_s \in \ZZ$ such that
%%
%\[
%\lim_{m \to \infty} \frac{1}{2m+1} \sum_{j=-m}^m \left( \prod_{k=1}^r V(j+n_k) -a \right) \cdot e^{-4 \pi i  \alpha j} \neq 0 \,.
%\]
%%
%The proof of Theorem~\ref{theorem1}(c) then tells us that in this case, there must exist a choice of $e_1, \ldots, e_r \in \{ \pm 1 \}$ such that
%%
%\[
%\alpha e_1P(X+n_1)^2 e_2(X+n_2)^2+\ldots +  e_kP(X+n_k)^2-2 \alpha X
%\]
%%
%is not a $W$-polynomial.

Since $P(X)= \alpha X^2+a_1X+a_0$, we have that $P(X+1)-P(X) -2\alpha X = \alpha+a_1$ is not a $W$-polynomial, and hence it is natural to try $n_1=0, n_2=1$. Let us calculate $\left( V(j) -a \right) \left( V(j+1) -a \right)  e^{-4 \pi i  \alpha  j}$ and see if it has non-zero mean. We have
\begin{align*}
&\left( V(j) -a \right) \left( V(j+1) -a \right)  e^{-4 \pi i  \alpha  j}= b^2 \cos(2 \pi P(j+1))\cos(2 \pi P(j))e^{-4 \pi i  \alpha  j}  \\
&= \frac{b^2}{4} e^{2 \pi i \left(P(j+1)+P(j)-2\alpha j\right)}  +\frac{b^2}{4} e^{2 \pi i \left( P(j+1)-P(j)-2\alpha j\right) }  \\
& \quad +\frac{b^2}{4} e^{2 \pi i \left(-P(j+1)+P(j)-2\alpha j\right)}  +\frac{b^2}{4} e^{2 \pi i \left(-P(j+1)-P(j)-2\alpha j\right)}.
\end{align*}
Now, each of the polynomials $P(X+1)+P(X) -2\alpha X, -P(X+1)-P(X)-2\alpha X$ has irrational quadratic coefficient and hence it is a $W_2$-polynomial. Also, $-P(X+1)+P(X)-2\alpha X=-4 \alpha X -\alpha -a_1$ is a $W_1$-polynomial. Theorem~\ref{Thm-law} then implies
\begin{align*}
M(e^{2 \pi i \left(P(j+1)+P(j)-2\alpha j\right)}) &=M( e^{2 \pi i \left(-P(j+1)+P(j)-2\alpha j\right)}) \\
&=M(e^{2 \pi i \left(-P(j+1)-P(j)-2\alpha j\right)})=0.
\end{align*}
Therefore, as $P(j+1)-P(j) -2\alpha X= \alpha+a_1$, we have
\begin{align*}
\lim_{m \to \infty} \frac{1}{2m+1} \sum_{j=-m}^m \left( V(j) -a \right) \left( V(j+1) -a \right)  e^{-2 \pi i  \alpha \cdot j}&= \frac{b^2}{4} M(e^{2 \pi i \left( P(j+1)-P(j)-2\alpha j\right) } ) \\
&= \frac{b^2}{4}e^{2 \pi i (\alpha+a_1)} \neq 0.
\end{align*}
This shows that $\beta=2\alpha$ is an eigenvalue.

\smallskip \noindent \textbf{(c)} Let  $n_1, \ldots, n_r, m_1, \ldots, m_s \in \ZZ$ be not necessarily distinct and let $k \in \ZZ$. As in Appendix~\ref{App:B}, we denote $F_n := E_n-a$. Then, $\AAA$ is the algebra generated by $\{1_{\XX(V)} \} \cup \{ F_n : n \in \ZZ\}$.

Let us calculate
\[
\Big\langle \prod_{j=1}^r F_{n_j}, T^k \prod_{l=1}^s F_{m_l} \Big\rangle.
\]
By unique ergodicity we have
\begin{align*}
&  \Big\langle \prod_{j=1}^r F_{n_j}, T^k \prod_{l=1}^s F_{m_l} \Big\rangle   = \lim_{m \to \infty} \frac{1}{2m+1} \sum_{j=-m}^m \prod_{p=1}^r   F_{n_p}(T^j V)  \prod_{l=1}^s F_{m_l+k}(T^j V) \\
&=\frac{ b^{r+s}}{2^{r+s}}  \lim_{m \to \infty}  \frac{1}{2m+1} \sum_{j=-m}^m  \prod_{p=1}^r \left( e^{2 \pi i P(j+n_p)} +e^{-2 \pi i P (j+n_p)} \right)  \prod_{l=1}^s \left( e^{2 \pi i P(j+m_l+k)}+e^{-2 \pi i P(j+m_l+k)}\right) \\
&=\frac{ b^{r+s}}{2^{r+s}}  \sum_{\ee \in M_{r+s}} \lim_{m \to \infty}  \frac{1}{2m+1} \sum_{j=-m}^m  e^{2 \pi i Q_\ee(j,k)},
\end{align*}
where $M_{r+s} \, := \, \{ \pm 1 \}^{r+s}$, and for each $\ee= (e_1, \ldots, e_r, f_1, \ldots, f_s) \in M_{r+s}$, we define
\[
Q_{\ee}(X,Y) \, := \, \left( \sum_{p=1}^r e_p P(X+n_p) \right)+  \left( \sum_{l=1}^s f_lP(X+Y+m_l) \right).
\]

A short computation gives
% \begin{align*}
% P(X+n_p)& =\alpha X^2 +(2 \alpha n_p + a_1) X+ \alpha n_p^2+ a_1 n_p+a_0\\
% P(X+Y+m_l)&=\alpha X^2+\alpha Y^2+2 \alpha XY + (2 \alpha m_l +a_1) X+ (2 \alpha m_l +a_1) Y +\alpha m_l^2+a_1m_l+a_0
% \end{align*}
% and
%
\begin{equation}\label{eq-QQ}
Q_{\ee}(X,Y)=  A_{\ee} \alpha X^2 + B_{\ee}\alpha  X + C_{\ee} \alpha  Y^2+2C_{\ee}\alpha   XY +  D_{\ee}  Y+ E_{\ee},
\end{equation}
where
\begin{align*}
A_{\ee} &=  e_1+e_2+e_3+ \ldots + e_r+f_1+f_2+ \ldots + f_s   \in \ZZ \\
B_{\ee} &= 2 \left(e_1n_1 +e_2n_2+ \ldots + e_rn_r +f_1m_1+\ldots +f_sm_s\right) +a_1 \cdot A_{\ee} \\
C_{\ee} &=  f_1+f_2+ \ldots + f_s  \in \ZZ
%   D_{\ee} &= 2\alpha \left(f_1m_1+\ldots +f_sm_s\right)  +a_1\cdot (f_1+\ldots+f_s)  \\
%   E_{\ee} &= \alpha \left(e_1n_1^2+e_2n_2^2+\ldots +e_rn_r^2 +f_1m_1^2+\ldots +f_sm_s^2\right)
%   &+ \left(e_1n_1 +e_2n_2+ \ldots + e_rn_r +f_1m_1+\ldots +f_sm_s\right) a_1 +A_{\ee} a_0
 \end{align*}
and $D_{\ee}, E_{\ee} \in \RR$.
Now, partition $M_{r+s}$ the following way:
\begin{align*}
  S_1 &:=\{ \ee \in M_{r+s} : A_{\ee} \neq 0 \}  \\
  S_2 &:=\{ \ee \in M_{r+s} : A_{\ee} =0, C_{\ee} \neq 0 \}  \\
  S_3 &:= \{ \ee \in M_{r+s} : A_{\ee} = C_{\ee} = 0, B_{\ee} \neq 0 \}  \\
  S_4 &:= \{ \ee \in M_{r+s} : A_{\ee} = C_{\ee} = B_{\ee} = 0 \} .
\end{align*}
Note that as a polynomial in $X$ we have
\[
Q_{\ee}(X,Y)=  A_{\ee} \alpha X^2 + (B_{\ee} +2 C_{\ee} Y) \alpha X  + (C_{\ee} \alpha Y^2+  D_{\ee} Y+ E_{\ee}).
\]
Now, for all $\ee \in S_1$, $Q_{\ee}(X,k)$ is a $W_2$-polynomial for all $k$. Moreover, for all $\ee \notin S_1$ we have $B_\ee \in \ZZ$. Thus, for all
$\ee \in S_3$, $Q_{\ee}(X,k)$ is a $W_1$-polynomial for all $k$. Therefore,
\begin{equation*}
\lim_{m \to \infty}  \frac{1}{2m+1} \sum_{j=-m}^m  e^{2 \pi iQ_\ee(j,k)} \,=\, 0 \qquad  \forall \ee \in S_1 \cup S_3, k \in \ZZ.
\end{equation*}
This gives
\begin{equation}\label{eq:S1}
\sum_{\ee \in S_1 \cup S_3} \lim_{m \to \infty}  \frac{1}{2m+1} \sum_{j=-m}^m  e^{2 \pi i Q_\ee(j,k)} \,=\, 0  \qquad \forall k \in \ZZ.
\end{equation}
%
%Next, for all $\ee \notin S_1$, we have
%
%\[
%P_{\ee}(X,Y)= (B_{\ee} +2C_{\ee} Y) \alpha X  + (C_{\ee} Y^2+  D_{\ee} Y+ E_{\ee}) \alpha   \,.
%\]
%
%Exactly as before, for all $\ee \in S_3$ and for all $k \in \ZZ$,  $Q_{\ee}(X,k)$ is a $W_1$-polynomial in $X$. Therefore,
%
%\begin{equation}\label{eq:S3}
%\sum_{\ee \in S_3} \lim_{m \to \infty}  \frac{1}{2m+1} \sum_{j=-m}^m  e^{2 \pi i Q_\ee(j,k)} \,=\, 0 \qquad  \forall k \in \ZZ \,.
%\end{equation}
%
Next, all $\ee \in S_2 \cup S_4$ we have $A_\ee=0$ and hence $B_\ee \in  \ZZ$. It follows that in this case, for each $k \in \ZZ$ we have that
\[
Q_{\ee}(X,k)=  (B_{\ee} +2 C_{\ee} k)\alpha  X  + (C_{\ee}\alpha k^2+  D_{\ee} k+ E_{\ee}) \alpha
\]
is linear polynomial in $X$ with leading coefficient in $\alpha \ZZ$.

Now,  for each $\ee \in S_2$, there exists at most one value $k_\ee$ of $k \in \ZZ$ such that $B_{\ee} +2\alpha C_{\ee} k_\ee=0$. Then, for all $k \neq k_\ee$, $Q_{\ee}(X,k)$ is a $W_1$-polynomial in $X$. Therefore,
\[
\lim_{m \to \infty}  \frac{1}{2m+1} \sum_{j=-m}^m   e^{2 \pi iQ_\ee(j,k)} \,=\, 0 \qquad  \forall \ee \in S_2, k  \neq  k_\ee.
\]
In particular, since $S_2$ is finite, there exists a finite set $F \subseteq \ZZ$ such that
\begin{equation*}
\lim_{m \to \infty}  \frac{1}{2m+1} \sum_{j=-m}^m   e^{2 \pi iQ_\ee(j,k)} \,=\, 0 \qquad  \forall \ee \in S_2, k \in \ZZ \backslash F.
\end{equation*}
Therefore, since all limits exist by Theorem~\ref{Thm-law} , there exists a function $u : \ZZ \to \CC$ supported inside $F$ such that
\begin{equation}\label{eq:S2}
\sum_{\ee \in S_2} \lim_{m \to \infty}  \frac{1}{2m+1} \sum_{j=-m}^m   e^{2 \pi iP_\ee(j,k)} \,=\, u(k) \qquad  \forall k \in \ZZ.
\end{equation}
Finally, for each $\ee \in S_4$ we have $Q_{\ee}(j,k)=  D_{\ee} k+ E_{\ee}$ is constant in $j$. Therefore, in this case
\begin{align*}
\lim_{m \to \infty}  \frac{1}{2m+1} \sum_{j=-m}^m   e^{2 \pi i Q_\ee(j,k)} \,=\, \lim_{m \to \infty}  \frac{1}{2m+1} \sum_{j=-m}^m    e^{2 \pi i \alpha (D_{\ee} k+ E_{\ee}) } =  e^{2 \pi i \alpha (D_{\ee} k+ E_{\ee}) }
\end{align*}
is a multiple of a character in $k$. Thus there exists a trigonometric polynomial $T$ such that
\begin{equation}\label{eq:S4}
\sum_{\ee \in S_4} \lim_{m \to \infty}  \frac{1}{2m+1} \sum_{j=-m}^m    e^{2 \pi i Q_\ee(j,k)} \,=\, T(k) \qquad  \forall k \in \ZZ.
\end{equation}
By combining \eqref{eq:S1}, \eqref{eq:S2} , and \eqref{eq:S4},  we have for all $n_1, \ldots, n_r, m_1, \ldots, m_s \in \ZZ$ not necessarily distinct,
\begin{align*}
\Big\langle \prod_{j=1}^r F_{n_j}, T^k \prod_{l=1}^s F_{m_l} \Big\rangle  & =  \frac{b^{r+s}}{2^{r+s}} \left( u(k)+T(k)\right).
\end{align*}
Now, let $f \in \AAA$ be arbitrary. Then, $f = c_01_{\XX(V)} + c_1 f_1+ \ldots + c_lf_l$ for some $f_1, \ldots, f_l$ of the form
\[
f_t= \prod_{j=1}^{r_t} F_{n_{j,t}}
\]
and $c_0, \ldots, c_l \in \CC$. Therefore,
\begin{align*}
  \langle f, T^k f \rangle &= |c_0|^2 + \sum_{t=1}^{l} c_0 \overline{c_t}   \langle 1_{\XX(V)}, T^k f_t \rangle  + \sum_{t=1}^{l} \overline{c_0} c_t   \langle f_t, 1_{\XX(V)} \rangle \\
  &+\sum_{s,t=1}^{l} c_s \overline{c_t}   \langle f_s, f_t \rangle = C+ \sum_{s,t=1}^{l} c_s \overline{c_t}   \langle f_s, T^k f_t \rangle
\end{align*}
for some constant $C$. Since by the above, for each $s,t$, the function $\langle f_s, T^k f_t \rangle $ is the sum of a function with finite support and a trigonometric polynomial, we get that
%can be written as a linear combination of $1_{\XX(V)}$ and} $\prod_{j=1}^r F_{n_j}$, using the linearity of $\langle \cdot , \cdot \rangle$ in each variable, we immediately get that
%
\[
\big\langle f, T^k f \big\rangle = v(k)+S(k)
\]
for some $v$ of finite support and trigonometric polynomial $S$.

Finally, noting that the diffraction spectrum is a.c. and non-trivial, the dynamical spectrum contains non-trivial a.c. spectrum. The claim follows.
\end{proof}

\begin{corollary}Let
\[
V(n) \,=\, \lambda \cos \Big( 2 \pi \Big( x_1 + n x_2 + \frac{n(n-1)}{2} \alpha \Big) \Big)
\]
for some $\lambda \neq 0$ and $\alpha \in \RR \backslash \QQ, x_1,x_2 \in \RR$. Then,
\begin{itemize}
  \item[(a)] $V$ is uniquely ergodic.
  \item[(b)] The set of eigenvalues of $\XX(V)$ is $ \alpha \ZZ$ and all eigenvalues are topological.
  \item[(c)] The dynamical spectrum is mixed pure point and absolutely continuous.
\end{itemize}
\end{corollary}

\section{Concluding Remarks}

\subsection{Sparse Structures}

Consider a point set $\Lambda \subseteq \ZZ$ with the property that the difference of consecutive locations diverges. Then, observing that  \cite[Lemma~4.4]{HRS} and its proof still hold if we replace $\RR$ by $\ZZ$, the autocorrelation\footnote{If $\Lambda$ is too sparse, then autocorrelation means the so called counting autocorrelation, see \cite{HRS}. } $\gamma$ exists with respect to $A_n=[-n,n]$ and $\gamma= \delta_0$. In particular, as long as the difference of consecutive locations diverges, the diffraction spectrum is easy to calculate and is Lebesgue. This class of examples contains the cases $\Lambda = \{ a^n : n \in \NN \}$ for any $a>1$ (\cite[Example~4.5]{HRS}),  $\Lambda = \{ n! : n \in \NN \}$ for any $a>1$ (\cite[Example~4.6]{HRS}) and  $\Lambda = \{ f_n : n \in \NN \}$ (\cite[Example~4.7]{HRS}), where $f_n$ are the Fibonacci numbers.

%Here we can put the earlier discussion about sparse structures. We can point out that for a simple diffraction analysis, it suffices that the differences of consecutive locations diverge,

On the Schr\"odinger side, potentials of this type are more difficult to study. At least the absence of absolutely continuous spectrum is known in great generality thanks to Remling's Oracle Theorem \cite{R11}. Moreover, Zlatos \cite{Z04} was able to establish fractional dimensional estimates for spectral measures under suitable assumptions on the gaps.

An interesting example of a sparse model where the difference of consecutive locations does not diverge is the primes. Here the calculation of the diffraction is less trivial, but it follows from well known sieve estimates in number theory that the autocorrelation of the primes is also $\gamma= \delta_0$ \cite[Thm.~5.7]{HRS}. Again, on the Schr\"odinger side things are more complicated,
%In this context we can also briefly point to your work on the prime models since this is related to the skew-shift case in that
as there are essentially no rigorous results; see however \cite{dOP01} for some intriguing numerics for this model.

\subsection{Random Structures}

In general, randomness is associated with absolutely continuous diffraction spectrum. Many random structures show (mainly) absolutely continuous diffraction measure, and introducing randomness in a ordered structure often introduces an extra a.c. component to the diffraction diagram.

Below we discuss many models with a.c. diffraction spectrum and how introducing randomness affects the diffraction. For more details we recommend to the reader \cite{BBM} and \cite[Chapter 11]{TAO}.
Since many of the systems we look at live in $\RR$ or $\RR^d$, to be consistent with the literature, we will denote the autocorrelation measure by $\gamma$ (or $\gamma_{\Lambda}^{}$ or $\gamma_{\mu}^{}$ when we want to specify the corresponding point set or measure) and by $\reallywidehat{\gamma}$ (or $\reallywidehat{\gamma_{\Lambda}^{}}$ or $\reallywidehat{\gamma_{\mu}^{}}$) its diffraction measure, as a measure on $\RR$ or $\RR^d$ (see \cite[Chapter 9]{TAO} for definitions and details).

But before talking about these models, let us briefly discuss the ambiguity of the concept "a.c. diffraction spectrum".
Most of the time, by this it is understood that the measure $\widehat{\gamma}$ is absolutely continuous. On the other hand, Delone sets, and more generally positive measures with non-trivial diffraction, always have a "trivial" Bragg peak of positive intensity at the origin (which can be seen as the diffraction equivalent of the fact that $0$ is always an eigenvalue of the corresponding dynamical system). In this case, by "absolutely continuous diffraction", it is usually understood that, besides the Bragg peak in the origin, the rest of the diffraction measure is an a.c. measure. We will discuss both cases below.

To make things easy to follow, we introduced the concept of Bartlett diffraction spectrum in Definition~\ref{def-diff} (compare Remark~\ref{rem-diff}), and we will simply say that the diffraction is a.c.\ when $\widehat{\gamma}$ is an a.c.\ measure, and that the Bartlett diffraction is a.c.\ when $\widehat{\gamma}$ is the sum of a Bragg peak in the origin and an a.c.\ measure (or equivalently that $\Gamma$ is an a.c.\ measure).

%we use the following notation, see \cite[Remark~15]{BBM} for the motivation for this name.
%\begin{definition} Let $\mu$ be a translation bounded measure and let $\gamma_\mu$ be its autocorrelation with respect to some averaging sequence $A_n$. The measure
%\[
%\Gamma_{\mu}^{}:= \reallywidehat{\gamma_{\mu}^{}} - \left(\reallywidehat{\gamma_{\mu}^{}}(\{0\}) \right) \delta_0
%\]
%is called the \emph{Bartlett diffraction measure of $\mu$}.
%\end{definition}

The simplest method of creating a system with a.c.\ spectrum, or adding a.c.\ spectrum to a system, is via the so called \emph{Bernoullization} of the system.

Consider some potential $V : \ZZ \to \{ \pm 1 \}$. Let $(W_n)_{n \in \ZZ}$ be a family of i.i.d. random variables, taking the values $1, -1$ with probabilities $p$ and $1-p$. Define $u: \ZZ \to \CC$ via $u(n)=W_n V(n)$, which is a randomization of $V$. Then, we have

\begin{theorem}\cite[Lemma 11.2]{TAO}\label{thm-bern}
Let $V: \ZZ \to \{ \pm 1 \}$ be uniquely ergodic. Let $(W_n)_{n \in \ZZ}$ be a family of i.i.d. random variables, taking the values $1, -1$ with probabilities $p$ and $1-p$. Let $\gamma_{V}$ be the (unique) autocorrelation of $V\delta_{\ZZ}$ as a measure on $\RR$. Consider the \emph{Bernoullization} $u(n) := W_n V(n)$ of $V$. Then, almost surely, $u\delta_{\ZZ}$ has autocorrelation $\gamma_{u}$ and diffraction $\reallywidehat{\gamma_{u}}$ given by
$$
\gamma_{u} = (2p-1)^2 \gamma_{V} + 4p(1-p) \delta_0, \quad \reallywidehat{\gamma_{u}} = (2p-1)^2 \reallywidehat{\gamma_{V}} + 4p(1-p) \lambda.
$$
\end{theorem}

In particular, when $p=1-p=\frac{1}{2}$, the diffraction is a.c.\ no matter which $V$ we start with. The reason why the choice of $V$ does not matter is because in this case $u(n)=\pm 1$ with probability $\frac{1}{2}$ for all $n$.

Choosing $V(n)=1$ for all $n$, we get as a special case the pure Bernoulli case:

\begin{theorem}\cite[Proposition~11.1]{TAO}
 Let $(W_n)_{n \in \ZZ}$ be a family of i.i.d. random variables, taking the values $1, -1$ with probabilities $\frac{1}{2}$ and $\frac{1}{2}$. Let $u(n):=W_n$. Then,  almost surely, $u \delta_{\ZZ}$ has autocorrelation $\gamma_{u}=  \delta_0$ and diffraction $\reallywidehat{\gamma_{u}}= \lambda$.
\end{theorem}

On the Schr\"odinger side, random potentials have been studied in depth, both in the stationary and in the non-stationary case; see, for example, \cite{BDFGVWZ19, CKM87, GK25}.

Another interesting example with a.c.\ spectrum is given by random dimers on the line (see \cite[Example~11.4]{TAO}, \cite{BvE}). Partition $\ZZ$ into an arrangement of pairs of consecutive integers ("dimers"). Now, randomly assign to each dimer either the pair $(+1,-1)$ or $(-1, +1)$. Then, almost surely with respect to the natural ergodic measure which defines this process, the diffraction measure is $\sigma=(1-2 \cos(2 \pi x)) \lambda$, see \cite[Example~11.4]{TAO} and \cite[Proposition 2]{BvE}. Schr\"odinger spectral and quantum dynamical properties of the random dimer model have been investigated as well \cite{BG00, DT03, JS07, JSS03}. It is an interesting example of the phenomenon that dynamical localization may fail even though spectral localization holds.

As far as additional examples are concerned, random point processes on $\RR$ (see \cite[Section~3]{BBM} and \cite[Section~11.3]{TAO}), as well as the eigenvalues of certain random matrices (\cite[Section~11.3]{TAO}) lead to examples of random Delone sets whose Bartlett diffraction spectrum is an a.c.\ measure \cite[Theorem~11.2,Theorem~11.3]{TAO}, but they always have a Bragg peak at the origin.

\subsection{Schr\"odinger vs.\ Diffraction}

As already pointed out in the introduction, diffraction measures tend to become more regular with increased complexity/disorder, whereas Schr\"odinger spectral measures tend to become more singular with increased complexity/disorder. It is a very intrugiuing question whether these two apparent monotonicity properties are in some way linked, perhaps via some sort of duality.

Naively, the dream correspondence would be
$$
\mathrm{pp} \longleftrightarrow  \mathrm{ac}, \quad  \mathrm{sc} \longleftrightarrow  \mathrm{sc}, \quad  \mathrm{ac} \longleftrightarrow  \mathrm{pp}.
$$
However, this can already be seen to not hold in such simple terms, based on known results. For example,
\begin{itemize}

\item Almost periodicity ensures pp diffraction, while almost periodic Schr\"odinger operators can have pp, sc, or ac spectrum, even within a simple $1$-parameter family \cite{DF24}.

\item Sequences in the subshift generated by the period doubling substitution have pp diffraction \cite{TAO} and sc Schr\"odinger spectral measures \cite{DF24}, while some sequences in the subshift generated by the Thue-Morse substitution have sc diffraction \cite{TAO} and also sc Schr\"odinger spectral measures \cite{DF24} (note that the latter phenomenon is not even appropriately understood).

\end{itemize}

It would be of great interest to explore the apparent dual connection between the two spectral types in more detail.

\appendix

\section{Diffraction}\label{App:diff}

In general, as introduced by Hof \cite{Hof1}, diffraction theory is set up for so called translation-bounded measures in $\RR^d$ (or even more generally in a second countable locally compact Abelian groups $G$), using the Fourier theory of measures/distributions, convergence of measures with respect to the vague topology, and arbitrary averaging sequences (van Hove sequences). In this setting, we identify a bounded function $V$ on $\ZZ$ with the translation-bounded measure $\mu := \sum_{n \in \ZZ}^{} V(n) \delta_n$. Since we are working with functions/measures supported inside $\ZZ$, there is an equivalent much simpler approach to diffraction theory. Indeed, in this case, measures supported inside $\ZZ$ can be simply viewed as functions, translation-boundedness of the measure is equivalent to boundedness of the corresponding function, vague convergence of measures becomes pointwise convergence, and the Fourier theory is just the Bochner theorem for positive definite functions. With this in mind, we briefly introduce/review below diffraction theory for bounded functions on $\ZZ$ and refer the reader to \cite[Chapter~9]{TAO} for a brief review of diffraction in $\RR^d$.

Recall that $\ell^\infty(\ZZ)$ denotes the space of all bounded functions $f: \ZZ \to \CC$. We equip this space with the topology $\tau_{\mathsf{p}}$ of pointwise convergence. Also, for $R>0$ we denote by $B_{R} := \{ f \in \ell^\infty(\ZZ) : \| f \|_\infty \leq R \}$ the (closed) ball of radius $R$.
%On this ball, the topology $\tau_{\mathsf{p}}$ is given by the norm
%\[
%\| f \| \, := \, \sum_{n \in \ZZ} \frac{1}{2^{|n|}} \left| f(n) \right| \,.
%\]
A standard diagonalization argument shows that $(B_R, \tau_{\mathsf{p}})$ is compact space, and it is clear that $\tau_{\mathsf{p}}$ is metrizable on $B_R$.
%, and hence $(B_{R}, \tau_{\mathsf{p}})$ is a compact metrizable space.
%\begin{lemma}\label{lem-BRcomp}
%$(B_{R}, \tau_{\mathsf{p}})$ is a compact metrizable space.
%\end{lemma}

Recall that for a function $f: \ZZ \to \CC$, we denote by $\tilde{f} : \ZZ \to \CC$ the function $\tilde{f}(k) \,:= \, \overline{f(-k)}$ and that for $f,g \in \ell^1(\ZZ)$, the convolution $f*g$ is defined as $f*g (k) \,: =\, \sum_{j \in \ZZ} f(j) \cdot g(k-j)$.

Consider an arbitrary $V \in \ell^\infty(\ZZ)$ and some fixed F\o lner sequence $(F_m)_{m}$ in $\ZZ$. We will often use $F_m := \{- m, -m+1, \ldots, m-1, m \}$, but all the results we prove work for all F\o lner sequences. For each $m$, define
\[
\gamma_m \, := \, \frac{1}{\card(F_m)} V_m *\widetilde{V_m},
\]
where $V_m$ is the restriction of $V$ to $F_m$, $V_{m}\, := \,  V|_{F_m}$. As a bounded function of finite support, $V_m$ belongs to $\ell^1(\ZZ)$. By construction, $\gamma_m$ is a positive definite function. Let us note that for all $k \in \ZZ$, we have
\begin{align*}
 \gamma_m(k)  &=  \frac{1}{\card(F_m)}  V_m *\widetilde{V_m} (k)  = \frac{1}{\card(F_m)}  \sum_{j \in \ZZ} V_m(j) \cdot \widetilde{V_m} (k-j)   \\
 &=\frac{1}{\card(F_m)}  \sum_{j \in \ZZ} V_m(j) \cdot \overline{V_m (j-k)}.
\end{align*}
In particular,
\begin{align*}
\left| \gamma_m(k) \right| &\leq \frac{1}{\card(F_m)} \sum_{j \in \ZZ}  \left| V_m(j) \right| \cdot \| V \|_\infty \leq  \| V\|_\infty^2.
\end{align*}

This proves that $\gamma_m \in B_{\|V\|_\infty^2}$ for every $m$. Therefore, by compactness $\gamma_m$ has a subsequence $\gamma_{k_m}$ that converges pointwise to some (autocorrelation) $\gamma \in \ell^\infty(\ZZ)$.

As a pointwise limit of positive definite functions, any such limit $\gamma$ is positive definite, and therefore we can define the positive diffraction measure $\sigma$ via \eqref{eq-diff}.

The above discussion and Bochner's theorem (see for example \cite[Theorem~1.4.3]{Rud}) give

\begin{lemma}\label{lem-ac-exists} Let $V \in \ell^\infty(\ZZ)$ and let $(F_m)_{m}$ be a F\o lner sequence. Then, there exists  a subsequence of $(F_{k_m})_{m}$ of $(F_m)_{m}$ such that $\gamma_{k_m}$ converges pointwise to a positive definite function $\gamma$ on $\ZZ$. Moreover, there exists a finite positive measure $\sigma$ on $\TT$ such that for all $k \in \ZZ$, we have
\[
\gamma(k)= \int_{\TT} e^{2 \pi i \, k \cdot y} \, \dd \sigma(y+\ZZ).
\]
\end{lemma}

Let us now prove the alternate formula for the autocorrelation which we used in Proposition~\ref{prop-acdef}. While this formula is usually easier to work with, our above definition has the advantage that we get positive definiteness for free.

\begin{lemma}\label{lem-ac-alt} Let $V \in \ell^\infty(\ZZ)$ and let $(F_m)_{m}$ be a F\o lner sequence. For each $m$ set as above
\[
\gamma_m \, := \, \frac{1}{\card(F_m)} V_m *\widetilde{V_m}.
\]
Then, for all $k \in \ZZ$, we have
\[
\lim_{m \to \infty} \left( \gamma_m(k)- \frac{1}{\card(F_m)} \sum_{j \in F_m} V(j)\cdot \overline{V(j-k)} \right) \,= \, 0.
\]
In particular, the autocorrelation $\gamma$ exists with respect to $(F_m)_{m}$ if and only if $\gamma_m$ is pointwise convergent, and in this case we have
\[
\gamma(k)= \lim_{m \to \infty} \gamma_m(k) \qquad \forall k \in \ZZ.
\]
\end{lemma}

\begin{proof}
Let $k \in \ZZ$ be fixed but arbitrary.
%Then
%
%\begin{align*}
% \gamma_m(k) &
 %=\frac{1}{\card(F_m)} V_m *\widetilde{V_m}=\frac{1}{\card(F_m)} \sum_{j \in \ZZ}  V_m(j) \cdot \widetilde{V_m}(k-j)  \\
   %&=\frac{1}{\card(F_m)} \sum_{j \in \ZZ}  V_m(j)\cdot \overline{V_m(j-k)}
%   =\frac{1}{\card(F_m)} \sum_{j \in F_m}  V(j) \cdot \overline{V_m(j-k)}  \,.
%\end{align*}
%
Then,
\begin{align*}
  \left| \gamma_m(k)- \frac{1}{\card(F_m)} \sum_{j \in F_m} V(j) \cdot \overline{V(j-k)} \right| &= \frac{1}{\card(F_m)} \left| \sum_{j \in F_m} V(j) \cdot \left(\overline{V_m(j-k)} - \overline{V(j-k)}\right) \right| \\
  &\leq \frac{1}{\card(F_m)} \sum_{j \in F_m}  \left|V(j) \right| \cdot \left|V_m(j-k) -V(j-k) \right|.
\end{align*}
Now, by the definition of $V_m$ we have $\left|V_m(j-k) -V(j-k) \right| \neq 0 \Longrightarrow j-k \notin F_m \Longrightarrow j \notin k+F_m$. It follows that
\begin{align*}
  &\left| \gamma_m(k)- \frac{1}{\card(F_m)} \sum_{j \in F_m} V(j) \cdot \overline{V(j-k)} \right| \leq \frac{\|V\|_\infty}{\card(F_m)} \sum_{j \in F_m \backslash (k+F_m)}  \left|V_m(j-k) -V_m(j-k) \right|\\
  & \leq \frac{\card(F_m \backslash (k+F_m)) }{\card(F_m)} 2 \| V\|_\infty^2 \leq \frac{\card(F_m \triangle (k+F_m)) }{\card(F_m)} 2 \| V\|_\infty^2   \to 0
\end{align*}
by the F\o lner condition. The claims follow.
\end{proof}

Let us complete the section by recalling the following well-known result. For a proof, see for example \cite[Theorem~4.16 ]{LSS}.

\begin{corollary}\label{cor:trig-poly} Let
\[
V(n) = \sum_{j=1}^r c_k \cdot e^{2 \pi i y_k \cdot n}
\]
be a trigonometric polynomial. Then $V$ has unique autocorrelation $\gamma$ and diffraction $\sigma$ given by
$$
\gamma = \sum_{j=1}^r |c_k|^2 \cdot e^{2 \pi i y_k \cdot n}, \quad \sigma = \sum_{j=1}^r |c_k|^2 \delta_{y_k}.
$$
\end{corollary}

\section{Properties of $\XX(V)$}\label{App:B}

%\subsection{Unique ergodicity}

Here we review some of the properties of $\XX(V)$ for $V \in \ell^\infty(\ZZ)$. Many of the results in this section are particular cases of \cite{BL} for $G=\ZZ$.

Since amenability plays a central role in this section, let us start by stating the following immediate consequence of \cite[Remark~4.5.7]{MoSt} or \cite[Proposition 1.2]{LSS}:
\begin{lemma}\label{lem-amenab} Let $f \in \ell^\infty(\ZZ)$. Then, the following are equivalent:
\begin{itemize}
  \item[(i)] $f$ is amenable.
  \item[(ii)] For all F\o lner sequences $(F_m)_{m}$, the limit
\[
 \lim_{m \to \infty} \frac{1}{\card(F_m)} \sum_{j \in l+F_m} f(j)
\]
exists uniformly in $l$.
  \item[(iii)] For all F\o lner sequences $(F_m)_{m}$, the limit
\[
 \lim_{m \to \infty} \frac{1}{\card(F_m)} \sum_{j \in F_m} f(j)
\]
exists and does not depend on $(F_m)_{m}$.
  \item[(iv)] For all F\o lner sequences $(F_m)_{m}$, the following limit exists,
\[
 \lim_{m \to \infty} \frac{1}{\card(F_m)} \sum_{j \in F_m} f(j).
\]
\end{itemize}
\end{lemma}

\begin{proof}
Let us first observe that in $\ZZ$ each compact set $K$ is finite and for all $A \subseteq \ZZ$, the $K$-boundary $\partial^{K}(A)$ (see \cite[Definition~1.1]{LSS}) satisfies $\partial^{K}(A) \subseteq \cup_{k \in K} A \triangle (k+A)$. This implies that in $\ZZ$, van Hove sequences and F\o lner sequences are the same thing. The claim follows from \cite[Proposition 1.2]{LSS}. Alternately, the claim follows from \cite[Remark~4.5.7]{MoSt} by mimicking the proof of \cite[Proposition 1.2]{LSS}.
\end{proof}

For $f:\XX(V) \to \CC$ and $w \in \XX(V)$, we denote by $f_{w}$ the function $f_{w} : \ZZ \to \CC$ given by $f_{w}(j) = f(T^j w)$. The following result is a well known variation of Oxtoby's Theorem (see for example \cite[Theorem~6.5]{LSS}).

\begin{proposition}
Let $V \in \ell^\infty(\ZZ)$. Then, $V$ is uniquely ergodic if and only if the set
$A := \{ f \in C(\XX(V)) : f_{V} \mbox{ is amenable } \}$ is dense in $C(\XX(V))$. Moreover, in this case, $f_w$ is amenable for all $f \in C(\XX(V))$ and all $w \in \XX(V)$.
\end{proposition}

By Stone--Weierstrass, the algebra generated by $\{1_{\XX(V)}\} \cup \{ E_n , \overline{E_n} : n \in \ZZ \}$ is dense in $C(\XX(V))$. We will denote this algebra by $\AAA$. Let us note that for each $a \in \CC$, $\AAA$ is also the algebra generated by $\{1_{\XX(V)}\} \cup \{ F_{n,a} , \overline{F_{n,a}} : n \in \ZZ \}$, where $F_{n,a}(w) \, := \, w(n)-a$.

Furthermore, if $V$ is real valued, we have $\overline{E_n} \,= \, E_n$ on $\XX(V)$. Therefore, $\AAA$ is the algebra generated by $\{1_{\XX(V)}\} \cup \{ E_n  : n \in \ZZ \}$. Similarly to the above, in this case $\AAA$ is also the algebra generated by $\{1_{\XX(V)}\} \cup \{ F_{n,a}  : n \in \ZZ \}$. Using the obvious relations $(f \cdot g)_{V} = (f_{V}) \cdot (g_{V})$, $(\overline{f})_{V} = \overline{f_{V}}$, $(E_n)_{V}(j) = E_n(T^j V)= (T^j V)(n)= V(n+j)$, we immediately get the following result:

\begin{theorem}\label{thm:uechar} Let $V \in \ell^\infty(\ZZ)$ and $a \in \CC$. Then,
\begin{itemize}
  \item[(a)] $V$ is uniquely ergodic if and only if for all not necessarily distinct $n_1, \ldots, n_r$, $m_1, \ldots, m_s \in \ZZ$, the function
  \[
  \ZZ \ni j \to \left( \prod_{k=1}^r V(j+n_k) -a \right) \cdot \overline{\left( \prod_{l=1}^s V(j+m_l)-a \right)}
  \]
  is amenable.
  \item[(b)] If $V$ is real valued, then $V$ is uniquely ergodic if and only if for all not necessarily distinct $n_1, \ldots, n_r \in \ZZ$, the function
  \[
  \ZZ \ni j \to \prod_{k=1}^r \left( V(j+n_k)-a \right)
  \]
  is amenable.
\end{itemize}
\end{theorem}

%\subsection{Continuity of Eigenfunctions}

Next we give a characterization for the topological eigenvalues in the case $\XX(V)$ is uniquely ergodic. Stronger versions of these results can be found in literature. We have the following lemma, which is essentially \cite[Proposition~3.5]{NS}.

\begin{lemma}\label{lem:eigenfunct}
Let $V \in \ell^\infty(\ZZ)$, $f \in C(\XX(V))$, and $\beta \in \TT$. Assume that $\ZZ \ni j \mapsto f(T^j V) \cdot e^{-2 \pi i  \beta \cdot j}$ is amenable. Then, $\ZZ \ni j \mapsto f(T^j w) \cdot e^{-2 \pi i  \beta \cdot j}$ is amenable for all $w \in \XX(V)$, and the mean
\[
F(w):= \lim_{m \to \infty} \frac{1}{2m+1} \sum_{j=-m}^m f(T^j w) \cdot e^{-2 \pi i  \beta \cdot j}
\]
is continuous and satisfies $F(Tw)=e^{2 \pi i  \beta } F(w)$ for all $w \in \XX(V)$. In particular, if $F(V)\neq 0$, the function $F$ is a continuous eigenfunction with eigenvalue $\beta$.
\end{lemma}

\begin{remark}
The mean of $ f(T^j V) \cdot e^{-2 \pi i  \beta \cdot j}$ is usually called the \emph{Fourier--Bohr coefficient of $f_{V}$ at $\beta$.}
\end{remark}

We conclude by proving the following result, which can be applied to the models we study in this paper.
\begin{proposition}\label{propC7}
Assume that $V \in \ell^\infty(\ZZ)$ is uniquely ergodic and let $a \in \CC$, $\beta \in \TT$. Assume that for all not necessarily distinct $n_1, \ldots, n_r, m_1, \ldots, m_s \in \ZZ$ the function
  \[
  \ZZ \ni j \to \left( \prod_{k=1}^r V(j+n_k) -a \right) \cdot \overline{\left( \prod_{l=1}^s V(j+m_l)-a \right)} \cdot e^{-2 \pi i  \beta \cdot j}
  \]
  is amenable.  Then, the following are equivalent:
\begin{itemize}
  \item[(i)] $\beta$ is an eigenvalue.
  \item[(ii)]$\beta$ is a topological eigenvalue.
  \item[(iii)] There exist not necessarily distinct $n_1, \ldots, n_r, m_1, \ldots, m_s \in \ZZ$ such that
  \[
  \lim_{m \to \infty} \frac{1}{2m+1} \sum_{j=-m}^m \left( \prod_{k=1}^r V(j+n_k) -a \right) \cdot \overline{\left( \prod_{l=1}^s V(j+m_l)-a \right)}\cdot e^{-2 \pi i  \beta \cdot j} \neq 0 .
  \]
\end{itemize}
\end{proposition}

\begin{proof}
{\bf (iii) $\Longrightarrow$ (ii)} follows from Lemma~\ref{lem:eigenfunct}.

\smallskip \noindent {\bf (ii) $\Longrightarrow$ (i)} is obvious.

\smallskip \noindent {\bf (i) $\Longrightarrow$ (iii)}: Let $f_\beta$ be an eigenfunction. If $\beta = 0$, there is nothing to prove, so we can assume that $\beta \neq 0$. Since $\AAA$ is dense in $C(\XX(V))$ and hence in $L^2(\XX(V), \mm)$, where $\mm$ is the unique ergodic measure on $\XX(V)$, there exists some $f \in \AAA$ such that $\langle f, f_\beta \rangle \neq 0$.
Since $\beta \neq 0$ we have $\langle 1_{\XX(V)}, f_\beta \rangle \,=\, 0$. Next, since $\AAA$ is spanned by
\[
\{1_{\XX(V)}\} \cup \{ F_{n_1,a} \cdot F_{n_2,a} \cdot \ldots F_{n_r,a} \cdot \overline{F_{m_1,a} \cdot F_{m_2,a} \cdot \ldots F_{m_s,a}} : n_1, \ldots n_r, m_1, \ldots, m_s \in \ZZ \},
\]
there exist not necessarily distinct $n_1, \ldots, n_r, m_1, \ldots, m_s \in \ZZ$ such that
\[
\langle F_{n_1,a} \cdot F_{n_2,a} \cdot \ldots F_{n_r,a} \cdot \overline{F_{m_1,a} \cdot F_{m_2,a} \cdot \ldots F_{m_s,a}} , f_\beta \rangle \neq 0.
\]
Therefore, by the ergodic theorem, for $\mm$-almost all $w \in \XX(V)$, we have
\begin{align*}
0 &\neq \langle F_{n_1,a} \cdot F_{n_2,a} \cdot \ldots F_{n_r,a} \cdot \overline{F_{m_1,a} \cdot F_{m_2,a} \cdot \ldots F_{m_s,a}} , f_\beta \rangle \\
&= \lim_{m \to \infty}  \frac{1}{2m+1} \sum_{j=-m}^m \left( \prod_{k=1}^r w(j+n_k) -a\right) \cdot \overline{\left( \prod_{l=1}^s w(j+m_l)-a \right)} \cdot e^{-2 \pi i  \beta \cdot j} .
\end{align*}
Finally, by Lemma~\ref{lem:eigenfunct} the limit
\[
F(w):=\lim_{m \to \infty} \frac{1}{2m+1} \sum_{j=-m}^m \left( \prod_{k=1}^r w(j+n_k)-a \right) \cdot \overline{\left( \prod_{l=1}^s w(j+m_l)-a \right)}\cdot e^{-2 \pi i  \beta \cdot j}
\]
exists for all $w \in \XX(V)$, is continuous, and satisfies $F(Tw)=e^{2 \pi i  \beta } F(w)$ for all $w \in \XX(V)$.
This implies that $|F|$ is a continuous function which is constant on $O(V)$ and hence on $\XX(V)$. Since $F(w) \neq 0$ for $\mm$-almost all $w \in \XX(V)$ we get
$F(w) \neq 0$ for all $w \in \XX(V)$. In particular, we have
\[
0 \neq F(V) = \lim_{m \to \infty} \frac{1}{2m+1} \sum_{j=-m}^m \left( \prod_{k=1}^r V(j+n_k) -a\right) \cdot \overline{\left( \prod_{l=1}^s V(j+m_l) -a\right)}\cdot e^{-2 \pi i  \beta \cdot j}.
\]
This completes the proof.
\end{proof}

\end{document}